\ifdefined\XeTeXversion\else\pdfoutput=1\fi
\documentclass[aps,prx,reprint,superscriptaddress,nofootinbib,longbibliography,floatfix]{revtex4-2}
\usepackage{float}
\usepackage[T1]{fontenc}
\usepackage{amsmath,amssymb,mathtools,bm}
\usepackage{amsthm}
\usepackage{array,booktabs}
\usepackage{graphicx}
\usepackage[protrusion=true,expansion=false]{microtype}
\usepackage{tikz}
\usetikzlibrary{arrows.meta,positioning,calc,fit,decorations.pathreplacing}
\usepackage{xcolor}
\definecolor{apsblue}{HTML}{185FA5}
\usepackage[colorlinks=true,citecolor=apsblue,linkcolor=apsblue,urlcolor=apsblue]{hyperref}

\newcommand{\cref}[1]{\autoref{#1}}
\newcommand{\Cref}[1]{\autoref{#1}}
\newtheorem{theorem}{Theorem}

\newtheorem{lemma}[theorem]{Lemma}
\newtheorem{corollary}[theorem]{Corollary}

\newcommand{\ket}[1]{\lvert #1\rangle}
\newcommand{\bra}[1]{\langle #1\rvert}
\newcommand{\braket}[2]{\langle #1\vert #2\rangle}
\newcommand{\proj}[1]{\lvert #1\rangle\!\langle #1\rvert}
\newcommand{\Tr}{\operatorname{Tr}}
\newcommand{\id}{\mathbb I}

\newcommand{\CCZ}{\operatorname{CCZ}}
\newcommand{\CZ}{\operatorname{CZ}}
\newcommand{\eps}{\varepsilon}
\newcommand{\cZ}{c_Z}
\newcommand{\cCCZ}{c_{\rm CCZ}}
\newcommand{\bits}{\{0,1\}}
\newcommand{\norm}[1]{\left\lVert #1\right\rVert}

\newcommand{\expect}[1]{\left\langle #1\right\rangle}

\begin{document}

\title{Robust Device-Independent Certification of Boolean-Phase Gates}

\author{Yunguang Han}
\email{hanyunguang@nuaa.edu.cn}
\affiliation{College of Computer Science and Technology, Nanjing University of Aeronautics and Astronautics, Nanjing 211106, People's Republic of China}

\author{Xingyuan Bu}
\affiliation{College of Computer Science and Technology, Nanjing University of Aeronautics and Astronautics, Nanjing 211106, People's Republic of China}

\author{Aleksandra Go\v{c}anin}
\email{aleksandra.gocanin@ff.bg.ac.rs}
\affiliation{Faculty of Physics, University of Belgrade, Studentski Trg 12-16, 11000 Belgrade, Serbia}

\author{Jiabing Yuan}
\email{jbyuan@nuaa.edu.cn}
\affiliation{College of Computer Science and Technology, Nanjing University of Aeronautics and Astronautics, Nanjing 211106, People's Republic of China}

\begin{abstract}
Device-independent certification of a quantum gate requires the input and
output tests to identify the same reference qubits. Self-testing the
output Choi state alone does not guarantee this consistency. We develop a
robust certification scheme for Boolean-phase gates, a broad family of
computational-basis diagonal gates specified by Boolean functions. Boolean derivatives convert the target-dependent phase information into
classical signs that can be evaluated from local measurement outcomes.
This leads to Bell tests built from CHSH blocks using two binary
measurements per party and no entangling measurements. At maximal
violation, the tests self-test the normalized Choi state and the measured
observables. Away from the maximum, they give an explicit affine lower
bound on the extracted-state squared fidelity that is uniform over all
Boolean functions and valid in arbitrary local dimensions. We then
combine an identity test and a gate-output test in an independent-source
network in which the reference devices use the same physical observables,
obtaining a closed-form Choi-fidelity bound for an effective \(n\)-qubit
channel.  For CCZ, a joint six-party analysis gives a stronger robustness
bound without changing the Bell expression or measurement settings. The
construction shows how the algebraic structure of a gate can shift
target-dependent information from quantum measurement design to classical
processing of local outcomes.
\end{abstract}

\maketitle

\section{Introduction}
\label{sec:intro}
Device-independent certification aims to characterize quantum devices from
observed correlations without assuming a model of the underlying states,
measurements, or Hilbert-space dimensions. For a quantum gate, the
normalized Choi state determines the channel once the input and output
reference systems are fixed
~\cite{Jamiolkowski1972,Choi1975}.  Self-testing the output Choi state alone,
however, does not ensure that an input test and a gate-output test identify
the same ideal reference qubits.  The two tests may select different
ideal-qubit embeddings, in which case the resulting state statements do not
refer to one physical operation.  A self-test of the output Choi state is
therefore not, by itself, a self-test of the physical gate.

Robust self-testing provides state-and-measurement guarantees directly from
Bell correlations, beginning with CHSH and extending to multipartite and
network settings
~\cite{Bell1964,CHSH1969,Tsirelson1980,PopescuRohrlich1992,
MayersYao2004,SupicBowles2020,McKagueYangScarani2012,
YangVertesiBancal2014,Bancal2015,BampsPironio2015,Kaniewski2016,
WangWuScarani2016,McKague2011,Wu2014,Li2018,Li2020,Baccari2020,
Coladangelo2017,ReichardtUngerVazirani2013,BalanzoJuando2026,
Supic2023,Liu2026Scalable}.  With characterized preparations and
measurements, quantum gates can also be verified efficiently using local
operations, including multiqubit controlled-phase gates
~\cite{ZhuZhang2020,Zeng2020,Zhang2022}.  In the fully
device-independent setting, compatible state certificates can be combined
into statements about channels, instruments, and memories
~\cite{DallArno2017,Sekatski2018,Wagner2020,Sekatski2023,Bock2024}.
More recently, exact device-independent certification has been established
for arbitrary unitary gates in independent-source networks
~\cite{Sarkar2026}, and self-testing has been extended to quantum supermaps
~\cite{Chiribella2008,Barizien2026}.  These results establish exact gate
certification in broad generality.  Here we address a different question:
whether a nontrivial structured family of multiqubit gates admits robust
gate-level guarantees while retaining a simple local measurement scheme.

We consider Boolean-phase gates,
\[
U_f=\sum_{x\in\{0,1\}^n}(-1)^{f(x)}\proj{x},
\]
where \(f:\{0,1\}^n\to\{0,1\}\) is an arbitrary Boolean function.  Modulo a
global phase, these are precisely the computational-basis diagonal
reflections, or equivalently arbitrary products of multiqubit
controlled-\(Z\) operations.  The family contains the non-Clifford CCZ gate
and higher controlled-phase reflections relevant to fault-tolerant
constructions~\cite{Jones2013,Eastin2013}, and its associated phase states
are closely related to hypergraph states
~\cite{Rossi2013,Gachechiladze2016,ZhuHayashi2019,MillerMiyake2016,
Takeuchi2019,ChenYanZhou2024,Huang2024}.  The key algebraic observation is
that Boolean differentiation converts the gate-dependent phase into a
\(Z\)-diagonal sign:
\begin{equation}
U_f X_i U_f^\dagger
=
X_i D_i^f(\bm Z_{\bar i}),
\label{eq:conjX}
\end{equation}
where the eigenvalue of \(D_i^f\) can be computed from local
\(Z\)-measurement outcomes at the other parties.  The gate dependence can
therefore be incorporated into the classical Bell scoring rule without
introducing additional or entangling measurements.

Using this observation, we construct a Bell inequality for every
Boolean-phase gate using two binary measurements per party.  Maximal
violation self-tests the normalized Choi state and the tested observables,
while nonmaximal scores give an explicit affine lower bound on the
extracted-state squared fidelity that applies uniformly across the family
and in arbitrary local dimensions.  We then combine an identity-channel
test and a gate-output test in an independent-source network, with the same
reference observables used in both tests.  Together with the
channel-composition framework of Ref.~\cite{Sekatski2018}, the two observed
Bell scores give a closed-form lower bound on the Choi fidelity of an
effective \(n\)-qubit channel.  The number of global route-and-setting
configurations grows linearly with \(n\).  For CCZ, the same Bell experiment
admits a stronger robustness bound when the three extracted pairs are
treated jointly.

The gate-level composition itself uses the established framework of
Ref.~\cite{Sekatski2018}; the new ingredient is the Boolean-derivative Bell
construction and its robust realization for an entire family of multiqubit
phase gates.  In particular, changing the target Boolean function changes
the classical scoring rule rather than the local measurement settings.
Our earlier work~\cite{Han2026CCZState} considered exact self-testing of the
single three-qubit state \(\CCZ\ket{+}^{\otimes3}\).  The present work
instead treats the \(2n\)-party Choi states of arbitrary Boolean-phase
gates, derives realization-independent robustness bounds, and promotes the
state tests to gate-level certification.  Sections~\ref{sec:bell} and
~\ref{sec:robust} establish the exact and robust state results,
Sec.~\ref{sec:gate} gives the gate-level certification, and
Sec.~\ref{sec:ccz} treats CCZ in detail.

\section{Bell inequalities and exact self-testing}
\label{sec:bell}

We first construct the Bell expression used to certify Boolean-phase Choi
states and establish its exact self-testing property.  The Boolean
derivative turns the phase appearing in each \(X\)-type Choi stabilizer
into a sign that can be evaluated from local \(Z\)-measurement outcomes.
The resulting test therefore uses only binary local measurements and
classical postprocessing.  Robust self-testing is developed in
Sec.~\ref{sec:robust}, and the state tests are combined into a gate
certificate in Sec.~\ref{sec:gate}.

\subsection{Boolean-phase gates and Choi states}
\label{subsec:choi}

Let \(f:\bits^n\to\bits\) be a Boolean function and define the corresponding
Boolean-phase gate by
\begin{equation}
U_f
=
\sum_{x\in\bits^n}
(-1)^{f(x)}
\proj{x}.
\label{eq:Uf}
\end{equation}
Adding the constant \(1\) to \(f\) changes \(U_f\) only by a global phase,
so we set \(f(0^n)=0\) whenever convenient.  Writing \(f\) in algebraic
normal form,
\[
f(x)
=
a_{\varnothing}
\oplus
\bigoplus_{\varnothing\ne S\subseteq[n]}
a_S\prod_{j\in S}x_j,
\qquad
a_S\in\bits,
\]
shows that, up to a global phase, \(U_f\) is an arbitrary product of
single-qubit \(Z\) gates and multiqubit controlled-\(Z\) gates associated
with the nonconstant monomials.  Thus Boolean-phase gates are precisely
the computational-basis diagonal reflections.  Since the \(2^n-1\)
nonconstant monomial coefficients can be chosen independently, the family
contains \(2^{2^n-1}\) distinct \(n\)-qubit unitary channels.

The \(i\)th Boolean derivative is
\begin{equation}
\partial_i f(x_{\bar i})
=
f(x_i=0,x_{\bar i})
\oplus
f(x_i=1,x_{\bar i}),
\label{eq:derivative}
\end{equation}
which is independent of \(x_i\).  For local \(Z\) reflections, define
\begin{equation}
D_i^f(\bm Z_{\bar i})
=
\sum_{x_{\bar i}\in\bits^{n-1}}
(-1)^{\partial_i f(x_{\bar i})}
\prod_{j\ne i}
\frac{\id+(-1)^{x_j}Z_j}{2}.
\label{eq:Di}
\end{equation}
The projectors in Eq.~\eqref{eq:Di} are orthogonal and complete, so
\(D_i^f\) is a \(Z\)-diagonal Hermitian reflection.  Since \(U_f\) is
diagonal,
\begin{equation}
U_f Z_i U_f^\dagger
=
Z_i,
\label{eq:conjZ}
\end{equation}
while Eq.~\eqref{eq:conjX} gives
\[
U_f X_i U_f^\dagger
=
X_iD_i^f(\bm Z_{\bar i}).
\]

Let
\begin{equation}
\ket{\Phi_d}
=
2^{-n/2}
\sum_{x\in\bits^n}
\ket{x}_A\ket{x}_B,
\qquad
d=2^n,
\label{eq:Phi}
\end{equation}
and define the normalized Choi state
\begin{equation}
\begin{aligned}
\ket{\Psi_f}
&=
(\id_A\otimes U_{f,B})\ket{\Phi_d}
\\
&=
2^{-n/2}
\sum_{x\in\bits^n}
(-1)^{f(x)}
\ket{x}_A\ket{x}_B .
\end{aligned}
\label{eq:PsiF}
\end{equation}
For each \(i\), define
\begin{align}
S_i^Z
&=
Z_{A_i}Z_{B_i},
\label{eq:SZ}\\
S_i^X
&=
X_{A_i}X_{B_i}
D_i^f(\bm Z_{B_{\bar i}}).
\label{eq:SX}
\end{align}
These \(2n\) reflections are obtained by conjugating the \(2n\)
independent stabilizers of \(n\) EPR pairs by
\(\id_A\otimes U_{f,B}\).  Their commutativity and independence therefore
follow immediately, and their unique common \(+1\) eigenstate is
\(\ket{\Psi_f}\).  The Boolean derivatives that describe the action of
\(U_f\) on local \(X\) operators thus also determine the \(X\)-type
stabilizers of its Choi state.

\subsection{Bell inequalities from Boolean derivatives}
\label{subsec:belltest}

The Bell experiment involves the \(2n\) separated parties
\(A_1,\ldots,A_n,B_1,\ldots,B_n\).  Each party has two binary
measurements.  We denote the two measurements of \(A_i\) by
\(A_{i,0},A_{i,1}\), and those of \(B_i\) by
\(Z_{B_i},X_{B_i}\).

Consider the term associated with the pair \(A_iB_i\).  Party \(B_i\)
measures \(X\), while every other \(B_j\), \(j\ne i\), measures \(Z\).
Let \(\xi_i\in\{\pm1\}\) be the \(X\)-measurement outcome of \(B_i\), and
let \(z_j\in\{\pm1\}\) denote the \(Z\)-measurement outcomes of the other
\(B_j\) parties.  The verifier converts these outcomes to bits
\[
u_j=\frac{1-z_j}{2}
\]
and computes
\begin{equation}
d_i
=
(-1)^{\partial_i f(u_{\bar i})},
\qquad
r_{i,1}
=
\xi_i d_i.
\label{eq:classicalscore}
\end{equation}
For the complementary \(Z\)-measurement contribution, set
\(r_{i,0}=z_i\).  Thus \(r_{i,0}\) and \(r_{i,1}\) are postprocessed
\(B\)-side binary outputs obtained entirely from the separated parties'
local outcomes.  In particular, \(D_i^f\) does not represent an additional
measurement setting: its eigenvalue is evaluated classically, and no joint
or controlled-phase measurement is performed.  The Boolean-derivative
scoring procedure is summarized in Fig.~\ref{fig:scoring}.

\begin{figure}[t]
\centering
\includegraphics[width=\columnwidth]
{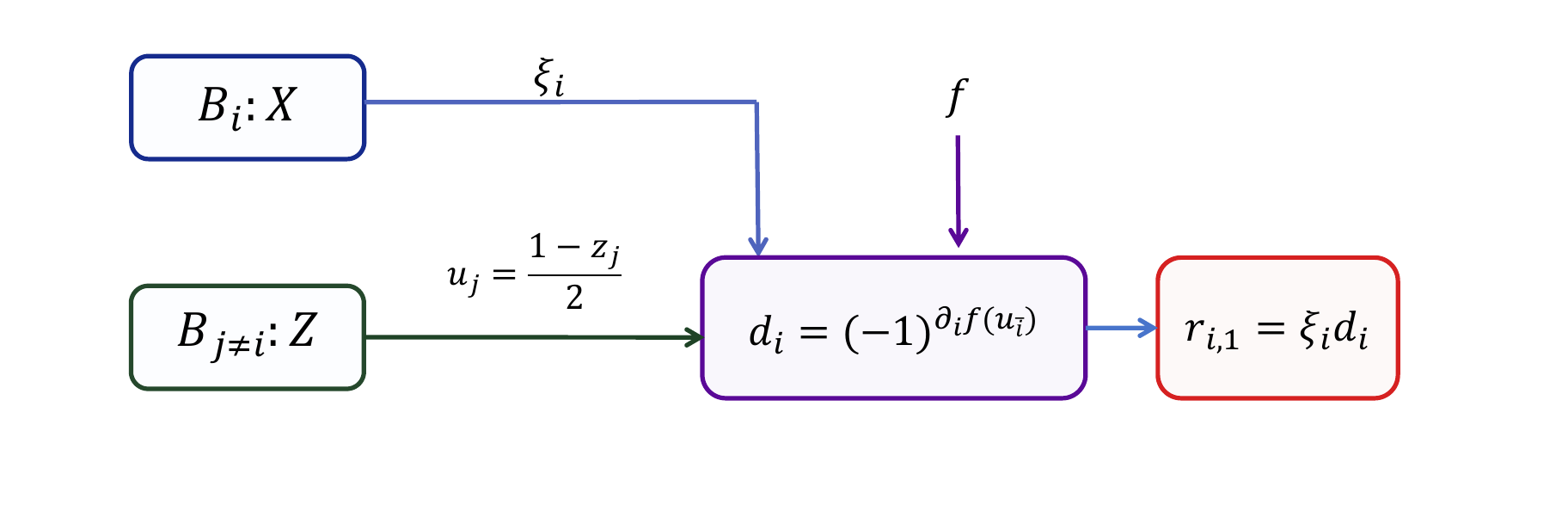}
\caption{
Boolean-derivative scoring for the \(i\)th Bell term.
Party \(B_i\) measures \(X\), while the remaining \(B_j\) parties measure
\(Z\).  The latter outcomes are mapped to bits
\(u_j=(1-z_j)/2\), from which the verifier evaluates the derivative sign
\(d_i=(-1)^{\partial_i f(u_{\bar i})}\).
Combining this sign with the \(X\)-measurement outcome \(\xi_i\) gives the
postprocessed output \(r_{i,1}=\xi_i d_i\).
}
\label{fig:scoring}
\end{figure}

We work in the tensor-product Bell model with unrestricted local
dimensions.  By purification and local Naimark dilation, it is sufficient
to consider a pure shared state \(\ket{\psi}\) and binary reflections.
For the operator analysis, the postprocessed outputs are represented by
the reflections
\begin{equation}
R_{i,0}
=
Z_{B_i},
\qquad
R_{i,1}
=
X_{B_i}
D_i^f(\bm Z_{B_{\bar i}}).
\label{eq:R}
\end{equation}
These operators are used only to analyze the Bell expression.
Operationally, all \(B_j\) parties remain spatially separated, and the
verifier combines their outputs only after the local measurements have
been performed.  Since \(D_i^f\) acts only on \(B_j\) with \(j\ne i\), it
commutes with both \(X_{B_i}\) and \(Z_{B_i}\), and \(R_{i,1}\) is itself
a reflection.

The \(i\)th CHSH operator is
\begin{equation}
\mathcal C_i^f
=
A_{i,0}(R_{i,0}+R_{i,1})
+
A_{i,1}(R_{i,0}-R_{i,1}),
\label{eq:block}
\end{equation}
and we define
\begin{equation}
\mathcal B_f
=
\sum_{i=1}^n
\mathcal C_i^f.
\label{eq:Bell}
\end{equation}
Eqs.~\eqref{eq:classicalscore}--\eqref{eq:Bell} define a Bell expression
linear in the observed conditional probabilities.

A sufficient set of global measurement settings grows linearly with \(n\).
On the \(B\) side, one uses the all-\(Z\) setting and the \(n\) settings in
which one \(B_i\) measures \(X\) while all other \(B_j\) measure \(Z\).
Each of these \(n+1\) settings is paired with the two uniform \(A\)-side
settings in which all \(A_i\) choose measurement \(0\) or all choose
measurement \(1\).  For the single-\(X\) setting associated with \(i\),
only the outcome of \(A_i\) enters the corresponding score term; the other
\(A\)-side outcomes can be ignored.  The Choi-state Bell test therefore
requires
\begin{equation}
2(n+1)=2n+2
\label{eq:statecontexts}
\end{equation}
global measurement-setting configurations.  This count concerns distinct
global settings; it does not by itself determine the number of estimated
score terms or the finite-sample complexity.  The additional configurations
required for gate certification are described in Sec.~\ref{sec:gate}.

As a simple example, for \(n\ge2\) consider the \(n\)-qubit
multi-controlled phase gate
\begin{equation}
C^{n-1}Z
=
\id-2\proj{1^n}.
\label{eq:CnZ}
\end{equation}
Its phase function is \(f(x)=\prod_{j=1}^n x_j\), for which
\begin{equation}
D_i^f
=
\id
-
2
\prod_{j\ne i}
\frac{\id-Z_{B_j}}{2}
=
C^{n-2}Z_{B_{\bar i}},
\label{eq:DiCnZ}
\end{equation}
where \(C^0Z=Z\).  The corresponding classical sign is
\begin{equation}
d_i
=
(-1)^{\prod_{j\ne i}u_j}.
\label{eq:andscore}
\end{equation}
Thus even for a multiqubit controlled phase, the Bell score requires only
local outcomes and classical evaluation of the Boolean derivative.  More
generally, the Walsh expansion of \(D_i^f\) may contain as many as
\(2^{n-1}\) \(Z\)-correlator monomials, whereas direct outcome
postprocessing evaluates the derivative sign without estimating these
terms separately.  The classical cost of this evaluation still depends on
the chosen representation of \(f\).  For \(n=3\),
Eq.~\eqref{eq:DiCnZ} gives the pairwise controlled-\(Z\) reflections
associated with CCZ, which is considered in detail in
Sec.~\ref{sec:ccz}.

\subsection{Exact self-testing}
\label{subsec:exact}

We use the standard notion of self-testing up to local isometries and
arbitrary auxiliary systems~\cite{SupicBowles2020}.  A maximal Bell
violation self-tests the target state and the action of the tested
observables if every realization attaining the maximum admits local
isometries that extract the ideal state and reproduce the ideal observables
on that state.

\begin{theorem}[Bell bounds and exact self-testing]
\label{thm:bellexact}
For every Boolean function \(f\),
\begin{equation}
\beta_{\mathrm L}(\mathcal B_f)
=
2n,
\qquad
\beta_{\mathrm Q}(\mathcal B_f)
=
2\sqrt2\,n.
\label{eq:bounds}
\end{equation}
The quantum bound is attained by \(\ket{\Psi_f}\) with
\begin{equation}
A_{i,0}
=
\frac{Z_{A_i}+X_{A_i}}{\sqrt2},
\qquad
A_{i,1}
=
\frac{Z_{A_i}-X_{A_i}}{\sqrt2},
\label{eq:Aideal}
\end{equation}
and Pauli \(Z_{B_i},X_{B_i}\).  Moreover, every realization attaining the
quantum maximum admits one local isometry per physical party and an
auxiliary state such that
\begin{equation}
\Phi\ket{\psi}
=
\ket{\Psi_f}\otimes\ket{\mathrm{aux}}.
\label{eq:exactextract}
\end{equation}
On the state, the extracted \(B_i\) observables act as Pauli \(Z\) and
\(X\), while the extracted \(A_i\) observables act as the rotated Pauli
measurements in Eq.~\eqref{eq:Aideal}.
\end{theorem}

For a deterministic fully local strategy, the local response functions
predetermine every postprocessed output \(r_{i,t}\).  For each fixed hidden
variable, the usual deterministic CHSH algebra therefore gives
\(\mathcal C_i^f\le2\).  The assignment
\(A_{i,0}=A_{i,1}=Z_{B_i}=1\) for every \(i\) attains \(2\) in each term,
independently of the remaining outputs and derivative signs.  Hence
\[
\beta_{\mathrm L}(\mathcal B_f)=2n.
\]

For the quantum bound, the \(B\) systems may be regarded as one composite
Hilbert space solely for the purpose of the operator estimate.
Since \(R_{i,0}\) and \(R_{i,1}\) are reflections and commute with the
\(A_i\) observables, Tsirelson's CHSH bound gives
\[
\mathcal C_i^f
\preceq
2\sqrt2\,\id.
\]
This grouping is used only for the operator bound and does not alter the
underlying \(2n\)-party Bell experiment.  For the ideal measurements in
Eq.~\eqref{eq:Aideal},
\begin{equation}
\mathcal C_i^f
=
\sqrt2\left(S_i^Z+S_i^X\right),
\label{eq:idealblock}
\end{equation}
and the generalized stabilizer relations attain \(2\sqrt2\) in every term.
It follows that
\(\beta_{\mathrm Q}(\mathcal B_f)=2\sqrt2\,n\).

To establish the self-testing statement, consider the CHSH
sum-of-squares identity
\begin{align}
2\sqrt2\,\id-\mathcal C_i^f
=
\frac1{\sqrt2}\bigg[
&\left(
A_{i,0}
-
\frac{R_{i,0}+R_{i,1}}{\sqrt2}
\right)^2
\nonumber\\
+&
\left(
A_{i,1}
-
\frac{R_{i,0}-R_{i,1}}{\sqrt2}
\right)^2
\bigg].
\label{eq:CHSHSOS}
\end{align}
If \(\mathcal B_f\) attains its quantum maximum, every
\(\mathcal C_i^f\) must attain \(2\sqrt2\).  The standard CHSH
self-testing relations then give local reflections
\(Z_{A_i},X_{A_i}\) satisfying, on the state,
\begin{align}
Z_{A_i}\ket{\psi}
&=
Z_{B_i}\ket{\psi},
\label{eq:exactrelZ}\\
X_{A_i}\ket{\psi}
&=
X_{B_i}
D_i^f(\bm Z_{B_{\bar i}})
\ket{\psi},
\label{eq:exactrelX}\\
\{X_{A_i},Z_{A_i}\}\ket{\psi}
&=
0,
\qquad
\{X_{B_i},Z_{B_i}\}\ket{\psi}
=
0.
\label{eq:exactac}
\end{align}

These relations also show why the \(n\) CHSH conditions identify one
global Boolean-phase state rather than \(n\) unrelated pairs.  The first
relation restricts the state to sectors with matching computational labels
on \(A_i\) and \(B_i\).  The second fixes the relative phase between
neighboring matched sectors \(x\) and \(x\oplus e_i\) to
\((-1)^{\partial_i f(x_{\bar i})}\).  These edge relations are mutually
consistent because all of the signs are derivatives of the same Boolean
function.  Along any path from \(0^n\) to \(x\), their product is
\[
(-1)^{f(x)\oplus f(0^n)},
\]
independently of the chosen path.  The resulting phase pattern is therefore
exactly that of \(\ket{\Psi_f}\), and a product of local SWAP isometries
extracts the target state and the tested observables.
Appendix~\ref{app:exact} gives the complete proof, including the
regularization, sector decomposition, and observable-extraction steps.

This proves exact self-testing at maximal violation.  The next section
derives a quantitative fidelity bound for the same Bell expression away
from the maximum.

\section{Robust self-testing}
\label{sec:robust}

Section~\ref{sec:bell} established exact self-testing at the quantum
maximum of \(\mathcal B_f\).  We now derive a quantitative fidelity bound
from a nonmaximal Bell score.  Throughout, \(F\) denotes the root Uhlmann
fidelity,
\[
F(\rho,\sigma)
=
\Tr\sqrt{\sqrt{\rho}\sigma\sqrt{\rho}},
\]
so that for a pure target,
\[
F^2(\rho,\proj{\psi})
=
\bra{\psi}\rho\ket{\psi}.
\]
No restriction is placed on the physical local Hilbert-space dimensions.

Let
\begin{equation}
\eps
=
2\sqrt2\,n-\expect{\mathcal B_f}
\ge0
\label{eq:deficit}
\end{equation}
be the total Bell deficit.

The robust analysis differs from an ordinary pairwise CHSH test because
\[
R_{i,1}
=
X_{B_i}D_i^f(\bm Z_{B_{\bar i}})
\]
contains observables from the other \(B\) systems.  An extraction channel
that changes those observables would in general also change the derivative
reflection \(D_i^f\).  We therefore choose the \(B\)-side extraction so
that every physical \(Z_{B_j}\) is preserved exactly.  Since
\(D_i^f\) is built only from the spectral projectors of these
\(Z\) observables, every Boolean-derivative reflection is preserved as
well.  The same CHSH operator bound can then be applied independently of
the particular Boolean function.

\begin{theorem}[Robust self-testing of Boolean-phase Choi states]
\label{thm:robustselftest}
For every Boolean function \(f\) and every quantum realization of
\(\mathcal B_f\), there exist product local extraction channels such that
the extracted \(2n\)-qubit state \(\rho_{\rm ext}\) satisfies
\begin{equation}
F^2(\rho_{\rm ext},\proj{\Psi_f})
\ge
1-\cZ\eps,
\qquad
\cZ
=
\frac{3(1+\sqrt2)}8.
\label{eq:robustF}
\end{equation}
\end{theorem}

The same coefficient \(\cZ\) is valid for every Boolean phase function
and for arbitrary local dimensions.  Its independence of \(n\) refers to
the slope with respect to the total Bell deficit, not to the behavior at
fixed noise as \(n\) increases.  For example, if each CHSH term has
expectation \(2\sqrt2\,v\), then
\[
\eps
=
2\sqrt2\,n(1-v).
\]
Maintaining a fixed global fidelity as \(n\) grows therefore requires
\[
1-v=O(1/n).
\]

\subsection{\(Z\)-preserving extraction channels}
\label{subsec:extraction}

We first state the properties of the extraction channels needed for the
robust bound.  Their blockwise construction using Jordan's lemma is given
in Appendix~\ref{app:robust}.  On every \(B_j\) block, the dual channel
satisfies
\begin{equation}
\Gamma_{B_j}^\dagger(Z_{B_j})
=
Z_{B_j}.
\label{eq:Zpreservationmain}
\end{equation}
It therefore also fixes the two spectral projectors
\((\id\pm Z_{B_j})/2\).  Since
\(D_i^f(\bm Z_{B_{\bar i}})\) depends only on the \(Z\) observables of the
other \(B\) systems, the product of their extraction channels obeys
\begin{equation}
\left(
\bigotimes_{j\ne i}\Gamma_{B_j}^\dagger
\right)
\!\left[
D_i^f(\bm Z_{B_{\bar i}})
\right]
=
D_i^f(\bm Z_{B_{\bar i}}).
\label{eq:Dpreservationmain}
\end{equation}

The extraction on \(A_i\) is the standard CHSH extraction determined only
by \(A_{i,0}\) and \(A_{i,1}\)~\cite{Kaniewski2016}.  Let
\(\Lambda^\dagger\) denote the product of all local dual extraction
channels.  For each \(i\), define the projector
\begin{equation}
P_i^f
=
\frac{(\id+S_i^Z)(\id+S_i^X)}4
\label{eq:pairprojector}
\end{equation}
and its pullback
\[
K_i^f
=
\Lambda^\dagger(P_i^f).
\]

Because Eq.~\eqref{eq:Dpreservationmain} leaves the derivative reflection
in \(S_i^X\) unchanged, its \(\pm1\) eigenspaces may be considered
separately.  In either eigenspace, the problem reduces to the same CHSH
extraction inequality.  This gives
\begin{equation}
K_i^f
\succeq
\id
-
\cZ
\left(
2\sqrt2\,\id-\mathcal C_i^f
\right).
\label{eq:pairtradeoffmain}
\end{equation}
The pairwise operator bound therefore depends only on the Bell deficit of
the \(i\)th CHSH term and not on the particular Boolean derivative appearing
in that term.  Appendix~\ref{app:robust} proves
Eq.~\eqref{eq:pairtradeoffmain} on the two-dimensional Jordan blocks and
extends the result to arbitrary local Hilbert-space dimensions.

\subsection{Robust fidelity bound}
\label{subsec:robustbound}

We now combine the pairwise inequalities into a fidelity bound for the
full Choi state.  Since the generalized Choi stabilizers commute, the
target projector factorizes as
\begin{equation}
P_f
:=
\proj{\Psi_f}
=
\prod_{i=1}^n P_i^f,
\label{eq:pairprojectors}
\end{equation}
where the \(P_i^f\) are defined in
Eq.~\eqref{eq:pairprojector}.  For commuting projectors,
\begin{equation}
\id-P_f
\preceq
\sum_{i=1}^n
(\id-P_i^f).
\label{eq:projectorunionmain}
\end{equation}

Set
\[
K_f
=
\Lambda^\dagger(P_f).
\]
Applying the positive unital map \(\Lambda^\dagger\) to
Eq.~\eqref{eq:projectorunionmain} and using
Eq.~\eqref{eq:pairtradeoffmain} gives
\begin{align}
\id-K_f
&\preceq
\sum_{i=1}^n(\id-K_i^f)
\nonumber\\
&\preceq
\cZ
\sum_{i=1}^n
\left(
2\sqrt2\,\id-\mathcal C_i^f
\right)
\nonumber\\
&=
\cZ
\left(
2\sqrt2\,n\,\id-\mathcal B_f
\right).
\label{eq:globaltradeoffmain}
\end{align}
Taking the expectation value in the physical state yields
\[
\Tr(\rho_{\rm ext}P_f)
=
\expect{K_f}
\ge
1-\cZ\eps.
\]
Since \(P_f=\proj{\Psi_f}\), this proves
Eq.~\eqref{eq:robustF}.

For \(\eps=0\), Theorem~\ref{thm:robustselftest} recovers the
state-extraction conclusion of Theorem~\ref{thm:bellexact}.  The particular
choice of extraction also provides the connection to gate certification.
The channel acting on each reference party \(A_i\) depends only on the two
physical observables \(A_{i,0}\) and \(A_{i,1}\).  Consequently, when the
same reference observables are used in the two test routes of
Sec.~\ref{sec:gate}, they define the same reference-side extraction.  This
shared reference map is the ingredient needed to combine the two state
tests into a statement about one effective channel.

\section{Device-independent gate certification}
\label{sec:gate}

Sections~\ref{sec:bell} and~\ref{sec:robust} certify the Choi state
associated with a Boolean-phase gate.  A gate statement requires one
additional ingredient: the input and output state tests must refer to the
same extracted reference system.  Following the channel-certification
framework of Ref.~\cite{Sekatski2018}, we impose this condition using an
independent-source network in which the choice between the input and
gate-output tests is hidden from the reference devices.  The same physical
reference observables are then used in both tests and, by Sec.~\ref{sec:robust},
define the same reference-side extraction.  Source independence plays a
separate role, allowing the resulting input interface to be chosen as a
product of local channels.

\subsection{Independent-source network}
\label{subsec:network}

For each \(i=1,\ldots,n\), an independent source \(S_i\) distributes an
arbitrary bipartite state on
\(\mathcal H_{A_i}\otimes\mathcal H_{Q_i}\).  Thus
\begin{equation}
\rho_{AQ}
=
\bigotimes_{i=1}^n
\rho_{A_iQ_i},
\label{eq:independence}
\end{equation}
with no restriction on the local dimensions.  Independent-source
structures of this type are standard in quantum-network nonlocality and
self-testing~\cite{Tavakoli2022,Renou2018,Bancal2018}.  The physical device
to be certified is a CPTP map
\begin{equation}
\mathcal E:
\mathsf L\!\left(
\bigotimes_{i=1}^n\mathcal H_{Q_i}
\right)
\longrightarrow
\mathsf L\!\left(
\bigotimes_{i=1}^n\mathcal H_{R_i}
\right).
\label{eq:physicalchannel}
\end{equation}
Any fixed internal ancilla of the physical process may be absorbed into
this CPTP description.  The independent-source model excludes
pre-existing correlations between such ancillary systems and the source
states in Eq.~\eqref{eq:independence}.

We use the usual Bell-model assumptions in each round.  The local
measurement inputs are chosen independently of the source variables, and
the spatially separated measurement stations do not communicate during the
round.  After the sources have emitted their systems, an independent route
choice \(e\in\{0,1\}\) is made.  For \(e=0\), the systems \(Q_i\) are sent
directly to the \(B_i\) stations.  For \(e=1\), they first pass through
\(\mathcal E\), and the output systems \(R_i\) are then sent to the same
stations.  The route choice is available to the routing device but is not
supplied to any reference party \(A_i\).  The resulting network is shown
schematically in Fig.~\ref{fig:network}.

\begin{figure}[t]
\centering
\includegraphics[width=\columnwidth]
{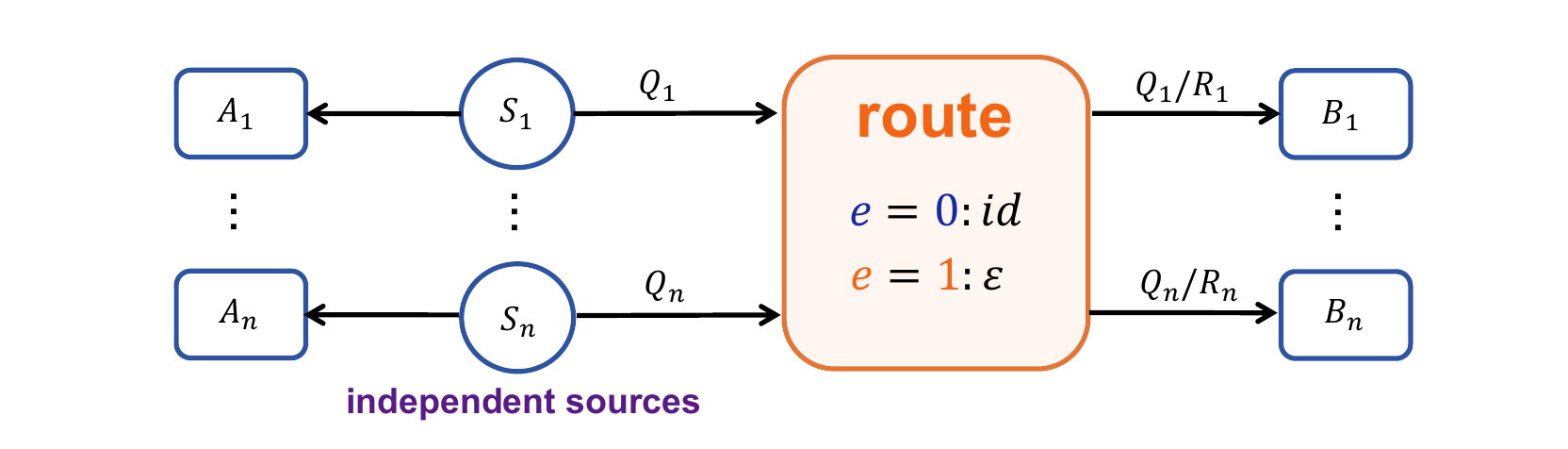}
\caption{
Independent-source network for gate certification.
Each source \(S_i\) distributes systems to the reference party \(A_i\)
and to the channel input \(Q_i\).
For \(e=0\), the \(Q_i\) systems are routed directly to the \(B_i\)
stations, whereas for \(e=1\) they pass through the physical channel
\(\mathcal E\), producing output systems \(R_i\).
The route choice \(e\) is not supplied to the reference parties.
}
\label{fig:network}
\end{figure}

Each \(A_i\) therefore receives only its Bell input and cannot condition
its measurements on \(e\).  The identity and gate-output tests consequently
use the same two physical observables \(A_{i,0}\) and \(A_{i,1}\).  Thus
their equality in the two tests is an operational restriction rather than
a choice of notation.

The two routes evaluate different Bell expressions.  The identity route
uses \(\mathcal B_0\), corresponding to the constant Boolean function
\(f=0\), while the gate-output route uses \(\mathcal B_f\).  Their Bell
deficits are
\begin{align}
\eps_{\rm in}
&=
2\sqrt2\,n
-
\langle\mathcal B_0\rangle_{e=0},
\label{eq:epsin}\\
\eps_{\rm out}
&=
2\sqrt2\,n
-
\langle\mathcal B_f\rangle_{e=1}.
\label{eq:epsout}
\end{align}

By Sec.~\ref{sec:robust}, the extraction on \(A_i\) is determined only by
the two physical observables \(A_{i,0}\) and \(A_{i,1}\).  Since these
observables are the same in the two routes, the corresponding state
self-tests use the shared reference-side extraction
\begin{equation}
\Lambda_A
=
\bigotimes_{i=1}^n
\Lambda_{A_i}.
\label{eq:referenceextraction}
\end{equation}
This is the condition needed to interpret the two state certificates as
statements about one effective channel.

For resource counting, a test configuration specifies both the route
\(e\) and one local measurement setting for every party.  The gate-output
route uses the \(2n+2\) global settings described in
Sec.~\ref{sec:bell}.  For the identity route, \(f=0\), so all Boolean
derivative signs are trivial.  It is sufficient to use the all-\(Z_B\)
and all-\(X_B\) settings, each paired with the two uniform \(A\)-side
settings.  The full gate-certification protocol therefore requires
\begin{equation}
2n+6
\label{eq:gateconfigurations}
\end{equation}
route-and-setting configurations.

\subsection{Effective channel and local interfaces}
\label{subsec:effectivechannel}

Device-independent certification does not identify the physical map
\(\mathcal E\) in an a priori calibrated basis.  Its action can instead be
characterized up to local input and output interfaces between the ideal
qubits and the physical systems.  Let
\[
\mathcal H_{\rm id}
=
(\mathbb C^2)^{\otimes n}
\]
denote the ideal \(n\)-qubit space.  We will show that there exist product
local interface channels
\begin{align}
\mathcal I
&=
\bigotimes_{i=1}^n\mathcal I_i,
&
\mathcal I_i:
\mathsf L(\mathbb C^2)
&\longrightarrow
\mathsf L(\mathcal H_{Q_i}),
\label{eq:injection}\\
\mathcal O
&=
\bigotimes_{i=1}^n\mathcal O_i,
&
\mathcal O_i:
\mathsf L(\mathcal H_{R_i})
&\longrightarrow
\mathsf L(\mathbb C^2),
\label{eq:extraction}
\end{align}
for which
\begin{equation}
\mathcal E_{\rm eff}
=
\mathcal O\circ\mathcal E\circ\mathcal I:
\mathsf L(\mathcal H_{\rm id})
\longrightarrow
\mathsf L(\mathcal H_{\rm id}).
\label{eq:Eeff}
\end{equation}
The existence of \(\mathcal I\) and \(\mathcal O\) is inferred from the
observed Bell correlations; they are not assumed to be calibrated
operations available in the experiment.  Accordingly, the certified
object is the effective channel \(\mathcal E_{\rm eff}\), rather than the
unrestricted action of \(\mathcal E\) outside the input subspace selected
by the certified interface.

Let
\[
\mathcal U_f(\cdot)
=
U_f(\cdot)U_f^\dagger
\]
denote the target unitary channel.  For fixed interfaces
\((\mathcal I,\mathcal O)\), define its Choi fidelity with the effective
channel as the root fidelity of their normalized Choi states,
\begin{equation}
F_{\rm Choi}(\mathcal E_{\rm eff},\mathcal U_f)
=
F\!\left[
(\id\otimes\mathcal E_{\rm eff})(\proj{\Phi_d}),
\proj{\Psi_f}
\right].
\label{eq:defchannelF}
\end{equation}
This is the channel fidelity used in the device-independent certification
framework of Ref.~\cite{Sekatski2018}, here with product local input and
output interfaces.

\subsection{Choi-fidelity bound}
\label{subsec:gatefidelity}

The robust self-test of Sec.~\ref{sec:robust} supplies extraction channels
for both test routes.  Their reference-side map is the shared
\(\Lambda_A\) of Eq.~\eqref{eq:referenceextraction}, while the channels
acting on the routed systems may differ.  Denote them by \(\Lambda_Q\) for
the identity route and by \(\Lambda_R\) for the gate-output route.  Define
\begin{align}
F_{\rm in}
&:=
F\!\left(
(\Lambda_A\otimes\Lambda_Q)(\rho_{AQ}),
\proj{\Phi_d}
\right),
\label{eq:Fin}\\
F_{\rm out}
&:=
F\!\left(
(\Lambda_A\otimes\Lambda_R)
[(\id_A\otimes\mathcal E)(\rho_{AQ})],
\proj{\Psi_f}
\right).
\label{eq:Fout}
\end{align}
Theorem~\ref{thm:robustselftest} gives
\begin{equation}
F_{\rm in}^2
\ge
1-\cZ\eps_{\rm in},
\qquad
F_{\rm out}^2
\ge
1-\cZ\eps_{\rm out}.
\label{eq:FinFout}
\end{equation}

\begin{corollary}[Choi-fidelity bound]
\label{cor:gatefidelity}
Under the Bell-model, timing, and source-independence assumptions above,
there exist product local channels \(\mathcal I\) and \(\mathcal O\) such
that
\begin{equation}
F_{\rm Choi}(\mathcal E_{\rm eff},\mathcal U_f)
\ge
\max\!\left\{
0,\,
\cos\!\left[
\arccos F_{\rm in}
+
\arccos F_{\rm out}
\right]
\right\}.
\label{eq:channelF}
\end{equation}
In particular, whenever
\(\cZ\eps_{\rm in}\le1\) and
\(\cZ\eps_{\rm out}\le1\), the observed Bell scores imply
\begin{align}
F_{\rm Choi}(\mathcal E_{\rm eff},\mathcal U_f)
\ge
\max\Bigl\{0,\,
\cos\bigl[
&\arccos\sqrt{1-\cZ\eps_{\rm in}}
\nonumber\\
&+
\arccos\sqrt{1-\cZ\eps_{\rm out}}
\bigr]
\Bigr\}.
\label{eq:channelDeficitBound}
\end{align}
\end{corollary}

\begin{proof}
The channel-composition theorem of Ref.~\cite{Sekatski2018} combines two
state certificates when they use the same extracted reference system.
Equation~\eqref{eq:referenceextraction} provides precisely this common
reference map in the present construction.  The route-dependent channels
\(\Lambda_Q\) and \(\Lambda_R\) provide the remaining state-extraction
maps.

Writing the Bures angle as
\[
A(\rho,\sigma)
=
\arccos F(\rho,\sigma),
\]
the composition theorem gives
\begin{equation}
A_{\rm Choi}
\le
\arccos F_{\rm in}
+
\arccos F_{\rm out},
\label{eq:burescomposition}
\end{equation}
where
\[
A_{\rm Choi}
=
\arccos
F_{\rm Choi}(\mathcal E_{\rm eff},\mathcal U_f).
\]
The source-independence condition in
Eq.~\eqref{eq:independence} is not needed for this Bures-angle composition
itself.  Its role is to allow the fixed source marginals to be absorbed
port by port into the input interfaces, yielding the product channel
\(\mathcal I=\bigotimes_i\mathcal I_i\) in
Eq.~\eqref{eq:injection}.  Appendix~\ref{app:gate} gives this reduction
explicitly.

If
\[
\arccos F_{\rm in}
+
\arccos F_{\rm out}
\le
\frac{\pi}{2},
\]
taking the cosine of Eq.~\eqref{eq:burescomposition} gives the
nontrivial part of Eq.~\eqref{eq:channelF}.  If the angle sum exceeds
\(\pi/2\), the cosine is nonpositive and the trivial bound
\(F_{\rm Choi}\ge0\) is stronger.  This proves
Eq.~\eqref{eq:channelF}.  Substituting
Eq.~\eqref{eq:FinFout} then gives
Eq.~\eqref{eq:channelDeficitBound}.
\end{proof}

The shared reference extraction and source independence therefore play
different roles.  The former ensures that \(F_{\rm in}\) and
\(F_{\rm out}\) refer to the same ideal reference qubits, while the latter
allows the input interface to have the product form
\(\mathcal I=\bigotimes_i\mathcal I_i\).  Without source independence,
the same state statistics can in general certify only a source-assisted
implementation whose input interface uses an unknown ancillary state
entangled across different input ports~\cite{Sekatski2018}.

At maximal Bell values,
\[
F_{\rm in}
=
F_{\rm out}
=
1.
\]
Both state certificates are then exact, and
Eq.~\eqref{eq:channelF} gives
\[
F_{\rm Choi}(\mathcal E_{\rm eff},\mathcal U_f)
=
1.
\]
The normalized Choi states therefore coincide, which implies
\[
\mathcal E_{\rm eff}
=
\mathcal U_f.
\]
The next section specializes the construction to CCZ and shows that its
additional structure gives a stronger robustness bound.

\section{The CCZ gate}
\label{sec:ccz}

CCZ is the canonical cubic Boolean-phase gate and a basic multiqubit
non-Clifford operation.  It provides a concrete instance of the general
gate-certification construction and, more importantly, a case in which the
relations among the Boolean derivatives can be used to improve the
robustness bound.  The Bell experiment itself remains unchanged.

\subsection{Bell inequality}
\label{subsec:cczbell}

For \(n=3\) and
\[
f(a,b,c)=abc,
\]
the Boolean-phase gate \(U_f\) is CCZ.  For each
\(i\in\{1,2,3\}\), let \(j\) and \(k\) denote the other two indices.
Eq.~\eqref{eq:Bell} then becomes
\begin{align}
\mathcal B_{\rm CCZ}
=
\sum_{i=1}^3
\bigl[
&A_{i,0}
\bigl(
Z_{B_i}
+
X_{B_i}\CZ_{B_jB_k}
\bigr)
\nonumber\\
&+
A_{i,1}
\bigl(
Z_{B_i}
-
X_{B_i}\CZ_{B_jB_k}
\bigr)
\bigr],
\label{eq:BCCZ}
\end{align}
where
\begin{equation}
\CZ_{B_jB_k}
=
\id
-
2P_{B_j}^1P_{B_k}^1,
\qquad
P_{B_j}^1
=
\frac{\id-Z_{B_j}}2.
\label{eq:CZscore}
\end{equation}
As in the general construction, the controlled-\(Z\) factor in
Eq.~\eqref{eq:BCCZ} does not represent an entangling measurement.  Its
eigenvalue is computed classically from the separated \(Z\)-measurement
outcomes of \(B_j\) and \(B_k\), and is \(-1\) only when both corresponding
bits are \(1\).

Theorem~\ref{thm:bellexact} gives
\begin{equation}
\beta_{\rm L}(\mathcal B_{\rm CCZ})
=
6,
\qquad
\beta_{\rm Q}(\mathcal B_{\rm CCZ})
=
6\sqrt2 .
\label{eq:CCZbounds}
\end{equation}
At the quantum maximum, the Bell correlations self-test the six-qubit
Choi state
\begin{equation}
\ket{\Psi_{\rm CCZ}}
=
\frac1{\sqrt8}
\sum_{a,b,c\in\{0,1\}}
(-1)^{abc}
\ket{abc}_A\ket{abc}_B
\label{eq:JCCZ}
\end{equation}
together with the tested observables.  Consequently, maximal violations
of \(\mathcal B_0\) and \(\mathcal B_{\rm CCZ}\) in the network of
Sec.~\ref{sec:gate} certify the CCZ gate exactly.

\subsection{Improved robustness bound}
\label{subsec:cczrobust}

The uniform robustness proof of Sec.~\ref{sec:robust} is designed to apply
to an arbitrary Boolean function.  It treats the derivative constraints
separately and uses a \(B\)-side extraction that preserves every physical
\(Z\) observable, and hence every
\(D_i^f(\bm Z_{B_{\bar i}})\), exactly.  This gives a coefficient
independent of \(f\), but does not use possible relations among the
derivatives of a particular gate.

For CCZ, these derivatives have the specific form
\[
D_1=\CZ_{B_2B_3},
\qquad
D_2=\CZ_{B_1B_3},
\qquad
D_3=\CZ_{B_1B_2}.
\]
They are therefore not three unrelated derivative reflections.  Instead of
bounding their contributions separately, we retain them jointly in the
full six-qubit extracted projector.  This permits the standard CHSH product
extraction~\cite{Kaniewski2016}, while the remaining nonfactorizing
six-party contribution is controlled by a single operator inequality.
Only the robustness analysis changes; the Bell expression and all
measurement settings remain the same.

Define the CCZ Bell deficit by
\begin{equation}
\eps_{\rm CCZ}
=
6\sqrt2
-
\expect{\mathcal B_{\rm CCZ}}.
\label{eq:cczdeficit}
\end{equation}

\begin{theorem}[Robust self-testing of the CCZ Choi state]
\label{thm:cczrobust}
For every quantum realization of \(\mathcal B_{\rm CCZ}\), there exist
product local extraction channels such that
\begin{equation}
\begin{aligned}
F^2(\rho_{\rm ext},\proj{\Psi_{\rm CCZ}})
&\ge
1-\cCCZ\eps_{\rm CCZ},
\\
\cCCZ
&=
\frac{4+5\sqrt2}{16}
\simeq
0.691942 .
\end{aligned}
\label{eq:cczrobust}
\end{equation}
\end{theorem}

For comparison, the uniform Boolean-phase coefficient of
Theorem~\ref{thm:robustselftest} is
\[
\cZ
=
\frac{3(1+\sqrt2)}8
\simeq
0.905330 .
\]
Thus
\[
\cCCZ<\cZ,
\]
and Eq.~\eqref{eq:cczrobust} gives a strictly stronger affine bound for the
same total Bell deficit.  The \(A_i\)-side extraction continues to depend
only on the measured observables \(A_{i,0}\) and \(A_{i,1}\).  The improved
bound is therefore compatible with the shared reference extraction required
in Sec.~\ref{sec:gate}.

The proof cannot be reduced to three independent pairwise inequalities.
After Jordan reduction, the full extracted CCZ projector leads to a
six-variable state-operator inequality.  Appendix~\ref{app:ccz} gives a
rigorous computer-assisted proof over the full Jordan domain, using exact
local estimates near the equality points and Bernstein-polynomial
certificates on the remaining region.  The coefficient \(\cCCZ\) is tight
for the particular product extraction and affine operator inequality used
there; no global optimality over all extraction channels or more general
robustness bounds is claimed.

For the identity route, the target factorizes into three Bell pairs.  The
same \(A\)-side extraction gives the corresponding bound with coefficient
\(\cCCZ\), as shown in Appendix~\ref{app:ccz}.  The two test routes
therefore satisfy
\begin{equation}
F_{\rm in}^2
\ge
1-\cCCZ\eps_{\rm in},
\qquad
F_{\rm out}^2
\ge
1-\cCCZ\eps_{\rm out}.
\label{eq:cczFinFout}
\end{equation}

\subsection{Fidelity bounds and visibility}
\label{subsec:cczvisibility}

Applying Corollary~\ref{cor:gatefidelity} to
Eq.~\eqref{eq:cczFinFout} gives
\begin{equation}
\begin{aligned}
&F_{\rm Choi}(\mathcal E_{\rm eff},\mathcal U_{\rm CCZ})
\\
&\quad\ge
\max\Bigl\{0,\,
\cos\bigl[
\arccos\sqrt{1-\cCCZ\eps_{\rm in}}
\\
&\hspace{8.2em}+
\arccos\sqrt{1-\cCCZ\eps_{\rm out}}
\bigr]
\Bigr\},
\end{aligned}
\label{eq:cczChannelDeficitBound}
\end{equation}
whenever
\(\cCCZ\eps_{\rm in}\le1\) and
\(\cCCZ\eps_{\rm out}\le1\).

For CCZ, the gate-output route uses eight of the test configurations
described in Sec.~\ref{sec:gate}: the four \(B\)-side setting patterns
consisting of the all-\(Z_B\) setting and the three single-\(X_B\)
settings, each paired with the two uniform \(A\)-side settings.  The
identity route requires four configurations.  Complete CCZ gate
certification therefore uses twelve route-and-setting configurations.

As a simple state-level noise model, consider the globally depolarized CCZ
Choi state
\begin{equation}
\rho_{v_{\rm dep}}
=
v_{\rm dep}\proj{\Psi_{\rm CCZ}}
+
(1-v_{\rm dep})\frac{\id}{64}.
\label{eq:white}
\end{equation}
With the ideal measurements,
\begin{equation}
\expect{\mathcal B_{\rm CCZ}}
=
6\sqrt2\,v_{\rm dep},
\label{eq:cczwhitebell}
\end{equation}
so the Bell inequality is violated for
\[
v_{\rm dep}>\frac1{\sqrt2}.
\]

For gate certification, we use instead a symmetric score-visibility model
in which the two observed Bell scores are reduced from their ideal values
by the same factor \(v\),
\[
\langle\mathcal B_0\rangle_{e=0}
=
\langle\mathcal B_{\rm CCZ}\rangle_{e=1}
=
6\sqrt2\,v.
\]
Here \(v\) parametrizes the observed Bell scores and does not assume the
depolarizing state model of Eq.~\eqref{eq:white}.

Figure~\ref{fig:cczrobustness} compares the fidelity guarantees obtained
from the uniform Boolean-phase bound and from the CCZ-specific joint
analysis.  For the uniform coefficient \(c_Z\), the fixed-extraction
state-fidelity bound becomes nontrivial for
\[
v>v_{\rm state}^{\rm uni}=0.869825,
\]
while the effective-channel Choi-fidelity bound becomes nontrivial for
\[
v>v_{\rm gate}^{\rm uni}=0.934913.
\]
Using the CCZ-specific coefficient \(c_{\rm CCZ}\) lowers these thresholds
to
\[
v>v_{\rm state}^{\rm CCZ}=0.829681
\]
and
\[
v>v_{\rm gate}^{\rm CCZ}=0.914840,
\]
respectively.  The higher channel thresholds arise because gate
certification combines two imperfect state certificates through the
Bures-angle bound.  Thus the joint CCZ analysis improves both the
state-level and gate-level guarantees without changing the Bell expression
or measurement settings.  Appendix~\ref{app:ccz} gives the corresponding
analytic expressions.

\begin{figure}[H]
\centering
\includegraphics[width=\columnwidth]{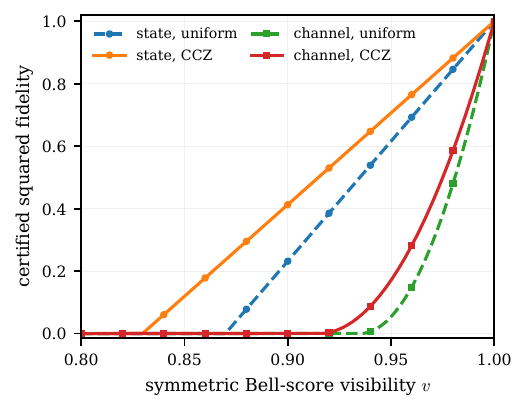}
\caption{
Comparison of the uniform Boolean-phase and CCZ-specific fidelity bounds
in the symmetric two-route score model.  The upper curves show the squared
Choi-state fidelity obtained from a single state test, while the lower
curves show the squared Choi fidelity of the effective channel obtained by
composing the input and output certificates.  Dashed curves use the uniform
coefficient \(c_Z\), and solid curves use the improved CCZ coefficient
\(c_{\rm CCZ}\).  The CCZ analysis lowers the nontrivial state-fidelity
threshold from \(0.869825\) to \(0.829681\) and the channel-fidelity
threshold from \(0.934913\) to \(0.914840\), without changing the Bell
expression or measurement settings.
}
\label{fig:cczrobustness}
\end{figure}

The CCZ case therefore illustrates two distinct features of the
construction.  The measurement architecture is inherited unchanged from
the general Boolean-phase test, while the additional relations among the
three CCZ derivatives allow a stronger conversion from Bell score to state
and channel fidelity.  The improvement requires neither additional
measurement settings nor entangling measurements.  Thus gate-specific
algebraic relations can strengthen the robustness guarantee without
changing the underlying Bell experiment.

\section{Conclusion and outlook}
\label{sec:discussion}

We have shown that Boolean-phase gates admit a simple and robust form of
device-independent certification.  The key observation is that Boolean
differentiation converts the gate-dependent phase structure into signs
that can be evaluated from local measurement outcomes.  This makes it
possible to use the same local binary measurement scheme throughout the
family, while changing only the classical scoring rule.  The resulting
Bell tests self-test the corresponding Choi states, provide explicit
robust fidelity guarantees, and can be combined in an independent-source
network to certify an effective quantum gate.  The CCZ case further shows
that additional algebraic relations among the derivatives can improve the
robustness without modifying the underlying Bell experiment.

The main structural feature of the construction is therefore a separation
between the quantum and classical parts of the certification task.  The
local measurements remain fixed, whereas the target operation enters
through classical processing of the observed outcomes.  This is especially
natural for Boolean-phase gates, which form a broad but highly structured
class of multiqubit phase operations containing CCZ and higher
controlled-\(Z\) gates.  The same viewpoint may be useful more generally
whenever the algebra of a quantum operation exposes relations that can be
tested through simple local observables rather than through more complex
measurement settings.

Two directions appear particularly natural.  First, Boolean-phase gates
use only phases \(0\) and \(\pi\), for which the relevant phase differences
reduce to binary signs.  Extending the construction to more general
diagonal gates would require a way to access nonbinary phase information
while retaining comparably simple local measurements.  Second, finite-data
bounds and more operational channel-distance guarantees would clarify how
the present robustness results translate into experimentally relevant
certification thresholds.  More broadly, the results suggest that known
algebraic structure in a quantum operation can be used to shift part of the
certification complexity from quantum measurement design to classical
processing of observed correlations.  Determining how far this principle
extends beyond Boolean-phase gates remains an open question.

\section*{Acknowledgments}
The authors thank Ivan \v{S}upi\'{c} for valuable discussions and insightful comments. The authors acknowledge support from the National Natural Science Foundation of China (Grants No. 62671305 and No. 62201252) and the Fundamental Research Funds for the Central Universities (Grants No. NS2025030 and No. NL2026011).

\section*{Data Availability}
The code and data are available in 
Ref.~\cite{BooleanPhaseCode2026}.

\section*{AUTHOR CONTRIBUTIONS}

Y.H. conceived and led the project, developed the main theoretical framework, and wrote the original manuscript. X.B. contributed to the theoretical analysis and derivations. A.G. contributed to the conceptual development, interpretation of the device-independent certification framework, and manuscript revision. J.Y. contributed to the conceptual discussions, supervision, and manuscript revision. All authors discussed the results and contributed to the final manuscript.

\appendix

\section{Exact self-testing proof}
\label{app:exact}

This appendix proves Theorem~\ref{thm:bellexact}.  We first derive the
relations implied by a maximally violated CHSH block, then reconstruct the
Boolean phase across the computational-basis sectors, and finally give the
local isometry extracting the target state and observables.  When analyzing
one CHSH block, we write \(A_0,A_1,R_0,R_1\) and suppress the party index.

\subsection{Relations from maximal CHSH violation}

At maximal violation, both squares in the sum-of-squares identity
Eq.~\eqref{eq:CHSHSOS} annihilate \(\ket{\psi}\).  Define
\begin{equation}
\widetilde Z_A
=
\frac{A_0+A_1}{\sqrt2},
\qquad
\widetilde X_A
=
\frac{A_0-A_1}{\sqrt2}.
\label{eq:rawaxes}
\end{equation}
The SOS relations give
\begin{equation}
\widetilde Z_A\ket{\psi}
=
R_0\ket{\psi},
\qquad
\widetilde X_A\ket{\psi}
=
R_1\ket{\psi}.
\label{eq:rawrel}
\end{equation}
Since \(A_0\) and \(A_1\) are reflections,
\begin{equation}
\{\widetilde Z_A,\widetilde X_A\}
=
0,
\qquad
\widetilde Z_A^2+\widetilde X_A^2
=
2\id.
\label{eq:rawaxisidentities}
\end{equation}
The standard CHSH regularization
~\cite{SupicBowles2020,Kaniewski2016} therefore gives anticommuting
reflections \(Z_A,X_A\) with the same action on the state:
\begin{equation}
Z_A\ket{\psi}
=
R_0\ket{\psi},
\qquad
X_A\ket{\psi}
=
R_1\ket{\psi},
\qquad
\{Z_A,X_A\}
=
0.
\label{eq:exactreg}
\end{equation}

Because the \(A\)-side operators commute with \(R_0\) and \(R_1\),
Eq.~\eqref{eq:exactreg} implies
\begin{equation}
\{R_0,R_1\}\ket{\psi}
=
0.
\label{eq:Ranticomm}
\end{equation}
For the block associated with \(A_iB_i\),
\[
R_0
=
Z_{B_i},
\qquad
R_1
=
X_{B_i}D_i^f(\bm Z_{B_{\bar i}}).
\]
Since \(D_i^f\) acts only on the other \(B\) parties, it commutes with
\(Z_{B_i}\) and \(X_{B_i}\), and hence
\[
\{R_0,R_1\}
=
\{Z_{B_i},X_{B_i}\}
D_i^f(\bm Z_{B_{\bar i}}).
\]
Using \((D_i^f)^2=\id\) in Eq.~\eqref{eq:Ranticomm} gives
\begin{equation}
\{Z_{B_i},X_{B_i}\}\ket{\psi}
=
0.
\label{eq:Banticomm}
\end{equation}
Together with Eq.~\eqref{eq:exactreg}, these are the state-dependent
relations
Eqs.~\eqref{eq:exactrelZ}--\eqref{eq:exactac} used in the main text.

\subsection{Computational-basis sectors and Boolean phases}

For each pair \(A_iB_i\), define
\begin{equation}
\begin{aligned}
P_{A_i}^{a}
&=
\frac{\id+(-1)^a Z_{A_i}}2,
\\
P_{B_i}^{b}
&=
\frac{\id+(-1)^b Z_{B_i}}2,
\\
a,b&\in\{0,1\}.
\end{aligned}
\label{eq:projectors}
\end{equation}
The relation
\(Z_{A_i}\ket{\psi}=Z_{B_i}\ket{\psi}\) implies
\(Z_{A_i}Z_{B_i}\ket{\psi}=\ket{\psi}\), and therefore
\begin{equation}
P_{A_i}^{a}P_{B_i}^{b}\ket{\psi}
=
0,
\qquad
a\ne b.
\label{eq:mismatch}
\end{equation}
Thus only sectors with matching computational labels contribute.

For \(x=(x_1,\ldots,x_n)\in\bits^n\), define
\begin{equation}
\ket{\psi_x}
=
\left(
\prod_{i=1}^n
P_{A_i}^{x_i}P_{B_i}^{x_i}
\right)
\ket{\psi},
\label{eq:branch}
\end{equation}
and denote the projector in parentheses by \(P_x\).  Also set
\begin{equation}
Y_i
:=
X_{A_i}X_{B_i}.
\label{eq:Yi}
\end{equation}
The \(Y_i\) commute for different \(i\).

We now make explicit how the relations of Sec.~\ref{subsec:exact} connect
neighboring computational-basis sectors.  From
Eq.~\eqref{eq:exactrelX}, multiplying by \(X_{B_i}\) gives
\[
Y_i\ket{\psi}
=
D_i^f(\bm Z_{B_{\bar i}})\ket{\psi}.
\]
Moreover, the anticommutation relations
Eqs.~\eqref{eq:exactac} and~\eqref{eq:Banticomm} imply, on the state, that
\(X_{A_i}\) and \(X_{B_i}\) exchange the two eigenspaces of
\(Z_{A_i}\) and \(Z_{B_i}\), respectively.  Since all projectors belonging
to other parties commute with \(Y_i\), we obtain
\[
Y_iP_x\ket{\psi}
=
P_{x\oplus e_i}Y_i\ket{\psi}.
\]
Combining the last two relations gives
\[
\begin{aligned}
Y_i\ket{\psi_x}
&=
P_{x\oplus e_i}
D_i^f(\bm Z_{B_{\bar i}})
\ket{\psi}
\\
&=
(-1)^{\partial_i f(x_{\bar i})}
P_{x\oplus e_i}\ket{\psi},
\end{aligned}
\]
where the second equality uses the fact that
\(D_i^f\) depends only on the bits other than \(x_i\).  Hence
\begin{equation}
Y_i\ket{\psi_x}
=
(-1)^{\partial_i f(x_{\bar i})}
\ket{\psi_{x\oplus e_i}}.
\label{eq:edgeflip}
\end{equation}

Since \(Y_i\) is unitary, Eq.~\eqref{eq:edgeflip} shows that adjacent
sectors have equal norm.  The matched sectors are mutually orthogonal and,
by Eq.~\eqref{eq:mismatch},
\[
\ket{\psi}
=
\sum_{x\in\bits^n}\ket{\psi_x}.
\]
Connectivity of the Boolean cube therefore gives
\[
1
=
\sum_{x\in\bits^n}\norm{\psi_x}^2
=
2^n\norm{\psi_{0^n}}^2,
\]
so
\begin{equation}
\norm{\psi_x}
=
2^{-n/2}
\qquad
\text{for every }x\in\bits^n.
\label{eq:branchnorm}
\end{equation}

It remains to determine the accumulated signs relating the different
sectors.  Choose any path
\[
x^{(0)}=0^n,\,
x^{(1)},\ldots,x^{(m)}=x
\]
along edges of the Boolean cube.  If the \(r\)th edge flips coordinate
\(i_r\), then
\[
\partial_{i_r}f
\bigl(x^{(r-1)}_{\bar i_r}\bigr)
=
f(x^{(r-1)})
\oplus
f(x^{(r)}).
\]
The derivative values therefore telescope:
\begin{equation}
\bigoplus_{r=1}^{m}
\partial_{i_r}f
\bigl(x^{(r-1)}_{\bar i_r}\bigr)
=
f(x)\oplus f(0^n).
\label{eq:phasetelescope}
\end{equation}
The accumulated sign depends only on the endpoint \(x\), not on the chosen
path.  Iterating Eq.~\eqref{eq:edgeflip}, and using the commutativity and
involutory property of the \(Y_i\), gives
\begin{equation}
\ket{\psi_x}
=
(-1)^{f(x)\oplus f(0^n)}
\left(
\prod_{i=1}^n
Y_i^{x_i}
\right)
\ket{\psi_{0^n}}.
\label{eq:pathintegral}
\end{equation}
Thus the matched computational-basis sectors reproduce the Boolean phase
pattern of the target Choi state.

\subsection{Local isometry and observable extraction}

For every physical party \(v\), let
\[
P_v^a
=
\frac{\id+(-1)^aZ_v}{2},
\qquad
a\in\{0,1\},
\]
and define the local isometry
\begin{equation}
\Phi_v\ket{\xi}
=
\ket{0}_vP_v^0\ket{\xi}
+
\ket{1}_vX_vP_v^1\ket{\xi}.
\label{eq:swap}
\end{equation}
Let
\[
\Phi
=
\bigotimes_v\Phi_v.
\]
Expanding the product isometry and using
Eq.~\eqref{eq:mismatch} removes all terms with unequal \(A_i\) and \(B_i\)
labels.  The remaining terms give
\begin{align}
\Phi\ket{\psi}
&=
\sum_{x\in\bits^n}
\ket{x}_A\ket{x}_B
\left(
\prod_{i=1}^nY_i^{x_i}
\right)
\ket{\psi_x}
\nonumber\\
&=
\sum_{x\in\bits^n}
(-1)^{f(x)\oplus f(0^n)}
\ket{x}_A\ket{x}_B
\ket{\psi_{0^n}}
\nonumber\\
&=
\ket{\Psi_f}\otimes\ket{\mathrm{aux}},
\label{eq:swapextract}
\end{align}
where
\begin{equation}
\ket{\mathrm{aux}}
=
(-1)^{f(0^n)}
2^{n/2}\ket{\psi_{0^n}}
\label{eq:auxstate}
\end{equation}
is normalized by Eq.~\eqref{eq:branchnorm}.  The constant
\(f(0^n)\) has therefore been absorbed into the auxiliary state.

The same isometry extracts the local observables.  Using the spectral
projectors in the definition of \(\Phi_v\), together with the corresponding
anticommutation relations, gives on \(\ket{\psi}\)
\begin{align}
\Phi Z_v\ket{\psi}
&=
\sigma_z^{(v)}
\ket{\Psi_f}\otimes\ket{\mathrm{aux}},
\label{eq:extractZ}\\
\Phi X_v\ket{\psi}
&=
\sigma_x^{(v)}
\ket{\Psi_f}\otimes\ket{\mathrm{aux}}.
\label{eq:extractX}
\end{align}
For the \(B_i\) parties these are the Pauli observables appearing in
Theorem~\ref{thm:bellexact}.  For the \(A_i\) parties,
Eqs.~\eqref{eq:rawaxes} and~\eqref{eq:exactreg} give
\begin{align}
\Phi A_{i,0}\ket{\psi}
&=
\frac{
\sigma_z^{(A_i)}+\sigma_x^{(A_i)}
}{\sqrt2}
\ket{\Psi_f}\otimes\ket{\mathrm{aux}},
\label{eq:extractA0}\\
\Phi A_{i,1}\ket{\psi}
&=
\frac{
\sigma_z^{(A_i)}-\sigma_x^{(A_i)}
}{\sqrt2}
\ket{\Psi_f}\otimes\ket{\mathrm{aux}}.
\label{eq:extractA1}
\end{align}
Eqs.~\eqref{eq:swapextract}--\eqref{eq:extractA1} complete the
state-and-measurement self-testing statement of
Theorem~\ref{thm:bellexact}.

\section{Proof of the robust self-testing bound}
\label{app:robust}

This appendix proves the pairwise operator inequality used in
Sec.~\ref{sec:robust} and completes the proof of
Theorem~\ref{thm:robustselftest}.  We first establish the bound on the
two-dimensional blocks supplied by Jordan's lemma, and then extend the
construction to arbitrary local Hilbert-space dimensions.

\subsection{Local extraction channels}

Apply Jordan's lemma to the two binary observables of each party.  Up to
local basis changes and outcome relabellings, a two-dimensional
\(A\)-side block can be written as
\begin{align}
A_0
&=
\cos a\,Z+\sin a\,X,
\nonumber\\
A_1
&=
\cos a\,Z-\sin a\,X,
\qquad
a\in[0,\pi/2],
\label{eq:robustAblocks}
\end{align}
while on a \(B\)-side block the two physical observables take the form
\begin{align}
\widehat Z_B
&=
Z,
\nonumber\\
\widehat X_B
&=
\cos(2b)\,Z+\sin(2b)\,X,
\qquad
b\in[0,\pi/2].
\label{eq:robustBblocks}
\end{align}
Here the hats distinguish the physical \(B\)-side observables from the
Pauli operators on the extracted qubit used below.  One-dimensional
Jordan blocks are obtained at the corresponding boundary points and may
be treated in the same way.

Define
\begin{equation}
g(t)
=
(1+\sqrt2)
\bigl(\sin t+\cos t-1\bigr),\\
0\le g(t)\le1
\quad
\text{for }t\in[0,\pi/2].
\label{eq:robustg}
\end{equation}
Let \(\Lambda_a^\dagger\) denote the dual \(A\)-side extraction channel.
Its action on the Pauli operators of the extracted qubit is
\begin{equation}
\begin{array}{c|ccc}
&X&Y&Z\\ \hline
0\le a\le\pi/4
&g(a)X&g(a)Y&Z\\
\pi/4\le a\le\pi/2
&X&g(a)Y&g(a)Z .
\end{array}
\label{eq:robustAmap}
\end{equation}
This is the standard CHSH extraction of
Ref.~\cite{Kaniewski2016}.  In the first angular regime it is a dephasing
channel about the \(Z\) axis, while in the second it is a dephasing channel
about the \(X\) axis.

For the \(B\) side, let \(\Gamma_b^\dagger\) be the dual channel defined
by
\[
\Gamma_b^\dagger(O)
=
\frac{1+g(b)}2\,O
+
\frac{1-g(b)}2\,ZOZ .
\]
Since \(0\le g(b)\le1\), this is a unital completely positive channel.  In
particular,
\begin{equation}
\Gamma_b^\dagger(X)=g(b)X,
\qquad
\Gamma_b^\dagger(Y)=g(b)Y,
\qquad
\Gamma_b^\dagger(Z)=Z.
\label{eq:robustBmap}
\end{equation}
The local basis changes that bring the physical observables into
Eqs.~\eqref{eq:robustAblocks} and~\eqref{eq:robustBblocks} are understood
as part of the corresponding extraction channels.

The essential feature of the \(B\)-side map is the exact identity
\[
\Gamma_b^\dagger(Z)=Z.
\]
It consequently fixes both spectral projectors
\((\id\pm Z)/2\).  When the map is applied to the physical \(B_j\) systems,
this property preserves every Boolean-derivative reflection appearing in
Eq.~\eqref{eq:Dpreservationmain}.  This is the step that allows the same
pairwise robustness estimate to be used independently of the particular
Boolean function.

\subsection{Pairwise operator inequality}

Let \(D\) be a Hermitian reflection acting on systems other than the local
systems \(A\) and \(B\), with
\[
[D,A_0]
=
[D,A_1]
=
[D,\widehat Z_B]
=
[D,\widehat X_B]
=
0.
\]
On the two extracted qubits, let \(Z_A,X_A,Z_B,X_B\) denote the target
Pauli operators.  Define
\begin{align}
P_D
&=
\frac14
(\id+Z_AZ_B)
(\id+X_AX_BD),
\label{eq:PD}\\
C_D
&=
A_0(\widehat Z_B+\widehat X_BD)
+
A_1(\widehat Z_B-\widehat X_BD).
\label{eq:CD}
\end{align}
Thus \(P_D\) is the derivative-conditioned target projector, whereas
\(C_D\) is the corresponding physical CHSH operator.

\begin{lemma}[Pairwise operator bound]
\label{lem:conditionalpair}
For the extraction channels in
Eqs.~\eqref{eq:robustAmap} and~\eqref{eq:robustBmap},
\begin{equation}
\begin{aligned}
&\bigl(
\Lambda_a^\dagger
\otimes
\Gamma_b^\dagger
\otimes
\operatorname{id}
\bigr)(P_D)
\\
&\quad\succeq
\id
-
\cZ
\bigl(
2\sqrt2\,\id-C_D
\bigr),
\qquad
\cZ
=
\frac{3(1+\sqrt2)}8 .
\end{aligned}
\label{eq:robustpair}
\end{equation}
\end{lemma}

\begin{proof}
We first resolve the Hilbert space into the \(D=\pm1\) eigenspaces.  In
either eigenspace, \(C_D\) is an ordinary CHSH operator; changing the sign
of \(D\) reverses the second physical \(B\)-side observable.  The two cases
are related within the same Jordan family.  Indeed, let
\[
\bar b
=
\frac{\pi}{2}-b.
\]
Then
\[
-\widehat X_B(b)
=
Z\,\widehat X_B(\bar b)\,Z,
\qquad
g(\bar b)=g(b),
\]
and the extraction channels satisfy
\[
\Gamma_{\bar b}^\dagger(ZOZ)
=
Z\Gamma_b^\dagger(O)Z.
\]
The target projector transforms by the same local Pauli conjugation.
Hence it is sufficient to establish Eq.~\eqref{eq:robustpair} in one
representative \(D\) eigenspace.

In that eigenspace, the pulled-back target projector and the CHSH operator
commute with the Hermitian involution
\[
(Z\otimes Z)(X\otimes X).
\]
Resolving its two eigenspaces, together with the two angular regimes of the
\(A\)-side extraction in Eq.~\eqref{eq:robustAmap}, reduces
Eq.~\eqref{eq:robustpair} to four scalar determinant inequalities depending
only on the Jordan angles \(a\) and \(b\).

The remaining positivity check is carried out by a rigorous
computer-assisted argument.  Using the tangent-half-angle parametrization
for \(a\) and \(b\), and multiplying by strictly positive denominators,
the four determinant conditions become bivariate polynomials with
coefficients in \(\mathbb Q(\sqrt2)\).  Their nonnegativity is established
on the full Jordan domain in two parts.  Near the points where the
polynomials vanish, exact Taylor bounds give the required sign.  The
compact complement is covered by finitely many rectangles with rational
endpoints.  On each rectangle, the corresponding polynomial is expressed
in the Bernstein basis after an affine rescaling of the variables.  All
Bernstein coefficients are verified to be nonnegative using exact
arithmetic over \(\mathbb Q(\sqrt2)\).

Because the Bernstein basis functions are nonnegative and form a partition
of unity, nonnegative Bernstein coefficients imply nonnegativity of the
polynomial throughout the rectangle.  The Taylor neighborhoods and the
finite rectangle cover exhaust the full parameter domain, proving
Eq.~\eqref{eq:robustpair}.  The explicit polynomials, exact local bounds,
rectangle cover, Bernstein coefficients, and verification scripts are
provided in the accompanying proof repository.
\end{proof}

For the Boolean-phase Bell test, set
\begin{equation}
D
=
D_i^f(\bm Z_{B_{\bar i}})
\label{eq:Dsubstitution}
\end{equation}
and identify \(A=A_i\), together with
\[
\widehat Z_B=Z_{B_i},
\qquad
\widehat X_B=X_{B_i}.
\]
Equation~\eqref{eq:Dpreservationmain} ensures that the extraction channels
on the other \(B\) systems leave
\(D_i^f(\bm Z_{B_{\bar i}})\) unchanged.  With these identifications,
\[
P_D
=
P_i^f,
\qquad
C_D
=
\mathcal C_i^f,
\]
and Lemma~\ref{lem:conditionalpair} gives precisely
Eq.~\eqref{eq:pairtradeoffmain}.

\subsection{Arbitrary local dimensions and tightness of the affine bound}

Jordan's lemma reduces a pair of binary reflections to invariant subspaces
of dimension at most two.  In finite dimensions this gives a direct-sum
decomposition, and the extraction channels above are applied separately on
each block while retaining the block label as a local auxiliary system.
Since Eq.~\eqref{eq:robustpair} holds on every block, it also holds for
their direct sum.

The same argument extends to infinite-dimensional local Hilbert spaces
through the corresponding direct-integral decomposition of the algebra
generated by two reflections.  The local channels act fiberwise as
functions of the Jordan parameter, and the fiber label is retained in the
auxiliary system.  The blockwise operator inequality therefore passes to
the direct integral.  Applying this construction independently at every
party gives the product extraction \(\Lambda\) used in
Sec.~\ref{sec:robust}.  Together with
Eqs.~\eqref{eq:projectorunionmain} and~\eqref{eq:pairtradeoffmain}, this
yields Eq.~\eqref{eq:globaltradeoffmain}, and hence
\[
F^2(\rho_{\rm ext},\proj{\Psi_f})
\ge
1-\cZ\eps .
\]
This proves Theorem~\ref{thm:robustselftest}.

We finally show that the coefficient \(\cZ\) cannot be decreased for the
chosen extraction channels within an affine operator inequality of the
form Eq.~\eqref{eq:robustpair}.  Consider the boundary point
\[
a=\frac{\pi}{2},
\qquad
b=0.
\]
There the pulled-back pair projector and CHSH operator reduce to
\begin{equation}
K_{a,b}
=
\frac{\id}{4},
\qquad
C_{a,b}
=
2X\otimes Z.
\label{eq:robustendpoint}
\end{equation}
On the \(+1\) eigenspace of \(X\otimes Z\), an affine inequality with
slope \(s\) requires
\begin{equation}
\frac14
\ge
1-s(2\sqrt2-2),
\label{eq:robustendpointbound}
\end{equation}
and therefore
\begin{equation}
s
\ge
\frac{3(1+\sqrt2)}8
=
\cZ.
\label{eq:robusttight}
\end{equation}
Thus the coefficient in Eq.~\eqref{eq:robustpair} is tight for the chosen
extraction within this affine state-operator bound.  This endpoint argument
does not establish optimality over all possible extraction channels or
over nonaffine robustness bounds.

\section{Proof of the gate-certification bound}
\label{app:gate}

This appendix completes the proof of
Corollary~\ref{cor:gatefidelity}.  The robust state bounds are supplied by
Theorem~\ref{thm:robustselftest}.  We first identify the two state tests
with the channel-composition construction of Ref.~\cite{Sekatski2018},
and then show how source independence allows the input interface to be
chosen as a product of local channels.

\subsection{Common reference extraction and channel composition}

Let
\[
\Lambda_A
=
\bigotimes_{i=1}^n\Lambda_{A_i}
\]
be the reference-side extraction used in both test routes, and let
\[
\Lambda_Q
=
\bigotimes_{i=1}^n\Lambda_{Q_i},
\qquad
\Lambda_R
=
\bigotimes_{i=1}^n\Lambda_{R_i}
\]
denote the extractions on the routed systems in the identity and
gate-output tests, respectively.  Define the corresponding extracted
states
\[
\rho_{\rm in}^{\rm ext}
=
(\Lambda_A\otimes\Lambda_Q)(\rho_{AQ})
\]
and
\[
\rho_{\rm out}^{\rm ext}
=
(\Lambda_A\otimes\Lambda_R)
\bigl[
(\id_A\otimes\mathcal E)(\rho_{AQ})
\bigr].
\]
By Eqs.~\eqref{eq:Fin} and~\eqref{eq:Fout}, their fidelities with the
ideal input and output Choi states are \(F_{\rm in}\) and \(F_{\rm out}\).

The correspondence with the channel-composition framework is then direct.
The systems \(A\) provide the common reference, \(Q\) and \(R\) are the
physical channel inputs and outputs, and
\[
\ket{\Psi_f}
=
(\id\otimes U_f)\ket{\Phi_d}
\]
is the ideal output Choi state.  The essential point is that the same
reference map \(\Lambda_A\) appears in both extracted states.  In the
network of Sec.~\ref{subsec:network}, this is enforced operationally:
the route choice is unavailable to the \(A_i\) devices, and each
\(\Lambda_{A_i}\) depends only on the same two physical observables
\(A_{i,0}\) and \(A_{i,1}\).

Define the Bures angle
\begin{equation}
A(\rho,\sigma)
=
\arccos F(\rho,\sigma).
\label{eq:buresangle}
\end{equation}
Applying the channel-composition result of Ref.~\cite{Sekatski2018} to
the two state certificates with their common reference extraction gives
input and output interfaces for which
\begin{equation}
\arccos
F_{\rm Choi}(\mathcal E_{\rm eff},\mathcal U_f)
\le
\arccos F_{\rm in}
+
\arccos F_{\rm out}.
\label{eq:bureschannel}
\end{equation}
The inequality follows from the triangle inequality for the Bures angle
together with its contractivity under quantum channels.  It is the step
that converts the two compatible state statements into a statement about
one effective channel.

If
\[
\arccos F_{\rm in}
+
\arccos F_{\rm out}
\le
\frac{\pi}{2},
\]
taking the cosine of Eq.~\eqref{eq:bureschannel} yields
\[
F_{\rm Choi}(\mathcal E_{\rm eff},\mathcal U_f)
\ge
\cos\!\left[
\arccos F_{\rm in}
+
\arccos F_{\rm out}
\right].
\]
If the angle sum exceeds \(\pi/2\), the right-hand side is nonpositive
and the trivial bound \(F_{\rm Choi}\ge0\) is stronger.  Hence
\begin{equation}
F_{\rm Choi}(\mathcal E_{\rm eff},\mathcal U_f)
\ge
\max\!\left\{
0,\,
\cos\!\left[
\arccos F_{\rm in}
+
\arccos F_{\rm out}
\right]
\right\},
\label{eq:bureschannelpositive}
\end{equation}
which is Eq.~\eqref{eq:channelF}.

Source independence is not required for the Bures-angle composition in
Eq.~\eqref{eq:bureschannel}.  Its role is instead to ensure that the input
interface appearing in the resulting channel statement can be chosen as a
product of local channels.

\subsection{Product input interfaces from source independence}

The input interface produced by the composition argument may use the fixed
physical source state as an ancillary resource.  For input port \(i\), the
corresponding local preprocessing may be represented by a CPTP map
\begin{equation}
\widetilde{\mathcal I}_i:
\mathsf L(\mathcal H_{Q_i}\otimes\mathbb C^2)
\longrightarrow
\mathsf L(\mathcal H_{Q_i}),
\label{eq:preinjection}
\end{equation}
where the second input is the ideal qubit to be injected and the first is
supplied with the fixed physical marginal
\[
\rho_{Q_i}
=
\Tr_{A_i}\rho_{A_iQ_i}.
\]

By the independent-source condition
Eq.~\eqref{eq:independence},
\begin{equation}
\rho_Q
=
\bigotimes_{i=1}^n\rho_{Q_i}.
\label{eq:productQmarginal}
\end{equation}
The ancillary state needed at each port can therefore be supplied locally
and absorbed into the corresponding preprocessing map.  Define
\begin{equation}
\mathcal I_i(\tau_i)
=
\widetilde{\mathcal I}_i
\bigl(
\rho_{Q_i}\otimes\tau_i
\bigr).
\label{eq:absorbedinjection}
\end{equation}
Since \(\rho_{Q_i}\) is fixed and normalized and
\(\widetilde{\mathcal I}_i\) is CPTP,
Eq.~\eqref{eq:absorbedinjection} defines a CPTP channel
\[
\mathcal I_i:
\mathsf L(\mathbb C^2)
\longrightarrow
\mathsf L(\mathcal H_{Q_i}).
\]
The complete input interface therefore has the product form
\begin{equation}
\mathcal I
=
\bigotimes_{i=1}^n\mathcal I_i.
\label{eq:productinjectionapp}
\end{equation}

The output interface is likewise the product of the local output
extractions,
\[
\mathcal O
=
\bigotimes_{i=1}^n\mathcal O_i.
\]
Together, these maps define the effective channel
\[
\mathcal E_{\rm eff}
=
\mathcal O\circ\mathcal E\circ\mathcal I
\]
of Eq.~\eqref{eq:Eeff}.

This product-interface interpretation is the point at which source
independence is needed.  If the physical input marginal \(\rho_Q\) were
entangled across different input ports, the ancillary state required by
the composition argument could not in general be prepared independently
inside the local maps \(\mathcal I_i\).  The same pair of state
certificates could then characterize a source-assisted implementation,
rather than an effective channel with product-local input interfaces.

Finally, Theorem~\ref{thm:robustselftest} gives
\[
F_{\rm in}
\ge
\sqrt{1-\cZ\eps_{\rm in}},
\qquad
F_{\rm out}
\ge
\sqrt{1-\cZ\eps_{\rm out}}
\]
whenever
\[
\cZ\eps_{\rm in}\le1,
\qquad
\cZ\eps_{\rm out}\le1.
\]
Substituting these bounds into
Eq.~\eqref{eq:bureschannelpositive} gives
Eq.~\eqref{eq:channelDeficitBound} and completes the proof of
Corollary~\ref{cor:gatefidelity}.

\section{Improved robustness bound for CCZ}
\label{app:ccz}

This appendix proves Theorem~\ref{thm:cczrobust}.  We use the standard
CHSH product extraction but retain the three CCZ derivative terms jointly.
After Jordan reduction, the problem becomes a six-parameter operator
inequality.  Its positivity over the full parameter domain is established
by a rigorous computer-assisted proof using exact local estimates and
finite Bernstein-polynomial certificates.  The complete machine-readable
certificates are provided in the accompanying verification repository.

\subsection{Joint CCZ operator bound}

By Jordan's lemma, up to local basis changes and outcome relabellings, the
two observables of each of the six parties can be written on every
two-dimensional block as
\begin{align}
M_0(\theta)
&=
\cos\theta\,Z+\sin\theta\,X,
\nonumber\\
M_1(\theta)
&=
\cos\theta\,Z-\sin\theta\,X,
\qquad
0\le\theta\le\frac{\pi}{2}.
\label{eq:cczjordan}
\end{align}
We denote the six Jordan angles by
\[
(a_1,a_2,a_3,b_1,b_2,b_3)
\in
[0,\pi/2]^6.
\]
One-dimensional Jordan blocks are included at the corresponding boundary
points.

Define
\begin{equation}
g(\theta)
=
(1+\sqrt2)
\bigl(\cos\theta+\sin\theta-1\bigr)
\label{eq:kaniewskig}
\end{equation}
and use the standard CHSH extraction of
Ref.~\cite{Kaniewski2016},
\begin{equation}
\begin{aligned}
\Lambda_\theta^\dagger(O)
&=
\frac{1+g(\theta)}2\,O
+
\frac{1-g(\theta)}2\,
\Gamma_\theta O\Gamma_\theta,
\\
\Gamma_\theta
&=
\begin{cases}
Z,&0\le\theta\le\pi/4,\\
X,&\pi/4\le\theta\le\pi/2.
\end{cases}
\end{aligned}
\label{eq:kaniewskich}
\end{equation}
Since \(0\le g(\theta)\le1\), each
\(\Lambda_\theta^\dagger\) is a unital completely positive map.  Unlike
the \(Z\)-preserving extraction used for the uniform Boolean-phase bound
in Appendix~\ref{app:robust}, the preserved Pauli axis here changes at
\(\theta=\pi/4\).  For CCZ, the resulting cross-pair terms are retained
and controlled jointly rather than bounded separately.

For the three extracted \(B\)-side qubits, it is convenient to use the
fixed output frame
\[
W_B
=
ZR_y(-\pi/4),
\qquad
R_y(\vartheta)
=
e^{-i\vartheta Y/2},
\]
and define
\[
\ket{\widetilde\Psi_{\rm CCZ}}
=
(\id_A\otimes W_B^{\otimes3})
\ket{\Psi_{\rm CCZ}}.
\]
This is only a fixed local change of output frame and can be absorbed into
the final extraction channels.  With
\[
\Lambda^\dagger
=
\bigotimes_{i=1}^3
\left(
\Lambda_{a_i}^\dagger
\otimes
\Lambda_{b_i}^\dagger
\right),
\]
define
\begin{equation}
\begin{aligned}
K_{\rm CCZ}
&=
\Lambda^\dagger
\bigl(\proj{\widetilde\Psi_{\rm CCZ}}\bigr),
\\
T_{\rm CCZ}
&=
K_{\rm CCZ}
-
\id
+
\cCCZ
\bigl(
6\sqrt2\,\id-\mathcal B_{\rm CCZ}
\bigr).
\end{aligned}
\label{eq:cczT}
\end{equation}
Here \(\mathcal B_{\rm CCZ}\) denotes the Bell operator restricted to the
current Jordan block.  It remains to prove
\[
T_{\rm CCZ}\succeq0
\]
for every point of the six-angle domain.

For the \(i\)th Bell block, let
\[
D_i
=
\CZ_{B_jB_k},
\]
where \(j\) and \(k\) are the two indices different from \(i\).  In the
two eigensectors \(D_i=\pm1\), the corresponding term is an ordinary CHSH
operator, with the second effective \(B_i\)-side observable reversed
between the two sectors.  These two cases are related within the same
extraction family.  Indeed, for
\[
\bar\theta
=
\frac{\pi}{2}-\theta,
\]
one has
\begin{equation}
\begin{aligned}
M_0(\bar\theta)
&=
HM_0(\theta)H,
\\
M_1(\bar\theta)
&=
-HM_1(\theta)H,
\\
\Lambda_{\bar\theta}^\dagger(HOH)
&=
H\Lambda_\theta^\dagger(O)H.
\end{aligned}
\label{eq:sectorcovariance}
\end{equation}
Moreover,
\[
W_B^\dagger H W_B=Z,
\]
so in the fixed output frame the extracted Bell-pair targets in the two
derivative sectors differ only by a local Pauli \(Z\) on Bob's qubit.

Let \(\widetilde K_i\) denote the direct sum over the two eigensectors of
\(D_i\) of the corresponding extracted Bell-pair target operators in this
frame.  The standard CHSH state-operator inequality
~\cite{Kaniewski2016} gives
\begin{equation}
T_i^{\rm eff}
=
\widetilde K_i
-
\id
+
\cCCZ
\bigl(
2\sqrt2\,\id-\mathcal C_i^{\rm CCZ}
\bigr)
\succeq0.
\label{eq:effectivechsh}
\end{equation}

The ideal CCZ projector can be written as a product of commuting
generalized-stabilizer projectors.  The dual extraction map, however, is
not multiplicative, so the three separately extracted pair operators do
not determine \(K_{\rm CCZ}\).  We isolate the remaining joint
contribution by defining
\begin{equation}
C_{1/2}
=
T_{\rm CCZ}
-
\frac12
\sum_{i=1}^3
T_i^{\rm eff}.
\label{eq:cczcoupling}
\end{equation}
It is therefore sufficient to establish
\[
C_{1/2}\succeq0
\]
throughout the six-angle domain.

The change of extraction axis in
Eq.~\eqref{eq:kaniewskich} partitions the domain into
\[
2^6=64
\]
closed angular regions.  More explicitly, for
\(\boldsymbol s=(s_1,\ldots,s_6)\in\{-,+\}^6\), let
\[
\mathcal R_{\boldsymbol s}
=
I_{s_1}\times\cdots\times I_{s_6},
\]
where
\[
I_-=[0,\pi/4],
\qquad
I_+=[\pi/4,\pi/2].
\]
On each \(\mathcal R_{\boldsymbol s}\), all six conjugating Pauli
operators \(\Gamma_\theta\) are fixed, so
\(C_{1/2}\) is analytic in the six Jordan angles.

On each angular region, we use the tangent-half-angle variables
\[
t_\ell
=
\tan\frac{\theta_\ell}{2},
\qquad
\ell=1,\ldots,6.
\]
After multiplication by a strictly positive common denominator,
\(C_{1/2}\) becomes a real symmetric \(64\times64\) matrix polynomial
whose entries lie in \(\mathbb Q(\sqrt2)\) and whose degree is at most two
in each variable.  Multiplication by the positive denominator does not
change positive semidefiniteness.

The equality configurations, where the matrix has vanishing eigenvalues,
are treated separately.  In neighborhoods of these points, exact Taylor
bounds together with Schur-complement estimates establish
positive semidefiniteness.  Removing these neighborhoods leaves a compact
subset of each angular region.  The remaining sets are covered by finitely
many dyadic boxes.

For each such box \(Q\), an affine rescaling maps the variables to
\([0,1]^6\), and the matrix polynomial is written in the tensor-product
Bernstein basis,
\[
\widehat C_Q(\bm t)
=
\sum_{\alpha}
B_{\alpha,Q}\,
b_{\alpha,Q}(\bm t),
\]
where
\[
b_{\alpha,Q}(\bm t)\ge0,
\qquad
\sum_\alpha b_{\alpha,Q}(\bm t)=1.
\]
Consequently,
\[
B_{\alpha,Q}\succeq0
\quad
\text{for every }\alpha
\]
implies
\[
\widehat C_Q(\bm t)\succeq0
\qquad
\text{throughout }Q.
\]

For completeness, the validation of each Bernstein coefficient matrix
uses a rigorous spectral-error criterion.  If a symmetric matrix
\(B_{\alpha,Q}\) has a certified approximation
\(\widetilde B_{\alpha,Q}\) satisfying
\[
\lambda_{\min}(\widetilde B_{\alpha,Q})
\ge
\ell_{\alpha,Q},
\qquad
\left\|
B_{\alpha,Q}
-
\widetilde B_{\alpha,Q}
\right\|_{\rm op}
\le
\eta_{\alpha,Q},
\]
with
\[
\ell_{\alpha,Q}
\ge
\eta_{\alpha,Q}
\ge0,
\]
then Weyl's inequality gives
\[
\lambda_{\min}(B_{\alpha,Q})
\ge
\ell_{\alpha,Q}-\eta_{\alpha,Q}
\ge0.
\]
The verification records supply the factorization data and rigorous bounds
needed for these inequalities.  All arithmetic associated with the exact
matrix-polynomial coefficients, box endpoints, and symbolic reductions is
performed over rational numbers and \(\mathbb Q(\sqrt2)\); rounding errors
in the matrix factorizations are enclosed by the stated operator-norm
bounds.

The full set of \(64\) angular regions is covered in the certificate.
Permutation symmetry among the three \(A_iB_i\) pairs and the covariance
in Eq.~\eqref{eq:sectorcovariance} are used only to reduce duplicate
computations.  For every angular region, the verification data record
either a direct certificate or an explicit symmetry transformation to a
certified representative region.  Since these transformations act by
parameter relabellings and unitary conjugations, they preserve positive
semidefiniteness.  Thus no angular region is omitted from the verification.

Combining the exact Taylor--Schur neighborhoods with the finite
Bernstein-certified box covers proves
\[
C_{1/2}\succeq0
\]
over the entire six-angle domain.  The exact matrix polynomials, the local
Taylor--Schur certificates, the complete angular-region table, the dyadic
box covers, the Bernstein coefficient matrices, and the validated
factorization records are supplied in the accompanying verification
repository.

It follows from Eq.~\eqref{eq:cczcoupling} and
Eq.~\eqref{eq:effectivechsh} that
\begin{equation}
T_{\rm CCZ}
=
C_{1/2}
+
\frac12
\sum_{i=1}^3
T_i^{\rm eff}
\succeq0.
\label{eq:cczglobalpositive}
\end{equation}

The same direct-integral argument as in Appendix~\ref{app:robust} extends
the blockwise inequality to arbitrary local Hilbert-space dimensions:
the Jordan-block extraction channels act fiberwise, while the block labels
are retained as local auxiliary systems.  Taking the expectation value of
Eq.~\eqref{eq:cczglobalpositive} in the physical state gives
\[
F^2
\bigl(
\rho_{\rm ext},
\proj{\widetilde\Psi_{\rm CCZ}}
\bigr)
\ge
1-\cCCZ\eps_{\rm CCZ}.
\]
Absorbing the fixed rotation \(W_B^{\otimes3}\) into the \(B\)-side output
maps gives
\[
F^2
\bigl(
\rho_{\rm ext},
\proj{\Psi_{\rm CCZ}}
\bigr)
\ge
1-\cCCZ\eps_{\rm CCZ},
\]
which proves Theorem~\ref{thm:cczrobust}.

\subsection{Tightness and visibility thresholds}

We first show that the coefficient \(\cCCZ\) cannot be decreased for the
chosen product extraction within the affine state-operator inequality used
above.  Consider a boundary configuration in which one extracted pair is
nonideal while the other two conditional CHSH pairs attain their ideal
values.  In the relevant matched sector, the extracted target weight is
\[
k_+
=
\frac{2+\sqrt2}{8}.
\]
The Bell deficit contributed by the nonideal pair is
\(2\sqrt2-2\).  Hence an affine inequality with slope \(s\) must satisfy
\[
k_+-1+s(2\sqrt2-2)\ge0.
\]
Therefore
\begin{equation}
s
\ge
\frac{1-k_+}{2\sqrt2-2}
=
\frac{4+5\sqrt2}{16}
=
\cCCZ.
\label{eq:ccztight}
\end{equation}
Thus \(\cCCZ\) is tight for the chosen extraction within this affine
state-operator family.  This argument does not establish optimality over
all extraction channels or over more general, including nonaffine,
robustness bounds.

For the identity route, the extracted target factorizes into three Bell
pairs.  If \(K_i\) denotes the extracted target operator for the \(i\)th
pair, then the \(K_i\) act on different extracted pairs and hence commute,
with
\[
0\preceq K_i\preceq\id.
\]
Therefore
\[
\id-\prod_{i=1}^3K_i
\preceq
\sum_{i=1}^3(\id-K_i).
\]
Applying the CHSH state-operator inequality
Eq.~\eqref{eq:effectivechsh} to the three ordinary CHSH pairs, and
Theorem~\ref{thm:cczrobust} to the gate-output route, gives
\begin{equation}
F_{\rm in}^2
\ge
1-\cCCZ\eps_{\rm in},
\qquad
F_{\rm out}^2
\ge
1-\cCCZ\eps_{\rm out}.
\label{eq:cczbranchboundsapp}
\end{equation}
The \(A_i\)-side extraction depends only on the same physical reference
observables reused in the two routes, so
Eq.~\eqref{eq:cczbranchboundsapp} is precisely the bound stated in
Eq.~\eqref{eq:cczFinFout}.

For a standalone CCZ Choi-state test, one may additionally use the constant
product extraction that outputs \(\ket{0}^{\otimes6}\).  Since
\[
\left|
\braket{0^6}{\Psi_{\rm CCZ}}
\right|^2
=
\frac18,
\]
this gives the realization-independent standalone bound
\begin{equation}
F_{\rm stand}^2
\ge
\max\left\{
\frac18,\,
1-\cCCZ\eps_{\rm CCZ}
\right\}.
\label{eq:robustFfloor}
\end{equation}
The \(1/8\) floor is not used in the gate certificate, because a constant
preparation does not provide the shared reference extraction required to
compose the input and gate-output state tests.

Finally, consider the symmetric score-visibility model used in
Sec.~\ref{sec:ccz},
\[
\langle\mathcal B_0\rangle_{e=0}
=
\langle\mathcal B_{\rm CCZ}\rangle_{e=1}
=
6\sqrt2\,v.
\]
Here \(v\) parametrizes the observed Bell scores and need not correspond
to the depolarizing parameter \(v_{\rm dep}\) introduced in
Eq.~\eqref{eq:white}.

For an affine robustness coefficient \(c\), define
\[
\delta_c(v)
=
6\sqrt2\,c(1-v).
\]
The fixed-extraction state and effective-channel bounds are
\begin{equation}
\begin{aligned}
F_{\rm state}^2(v;c)
&\ge
\max\left\{
0,\,
1-\delta_c(v)
\right\},
\\
F_{\rm Choi}^2(v;c)
&\ge
\max\left\{
0,\,
1-2\delta_c(v)
\right\}^2.
\end{aligned}
\label{eq:cczvisibilitycurves}
\end{equation}
The second line follows from the symmetric Bures-angle composition:
if
\[
F_{\rm in}^2,F_{\rm out}^2
\ge
1-\delta_c(v),
\]
then
\[
F_{\rm Choi}
\ge
\max\{0,1-2\delta_c(v)\}.
\]

Figure~\ref{fig:cczrobustness} evaluates
Eq.~\eqref{eq:cczvisibilitycurves} for the uniform Boolean-phase
coefficient \(c_Z\) and the CCZ-specific coefficient \(c_{\rm CCZ}\).
For a general coefficient \(c\), the corresponding thresholds are
\[
v_{\rm state}(c)
=
1-\frac{1}{6\sqrt2\,c},
\qquad
v_{\rm gate}(c)
=
1-\frac{1}{12\sqrt2\,c}.
\]
Thus
\begin{equation}
\begin{aligned}
v_{\rm state}^{\rm uni}
&=
1-\frac{1}{6\sqrt2\,c_Z}
=
0.8698252361\ldots,
\\
v_{\rm state}^{\rm CCZ}
&=
1-\frac{1}{6\sqrt2\,c_{\rm CCZ}}
=
0.8296805588\ldots,
\\
v_{\rm gate}^{\rm uni}
&=
1-\frac{1}{12\sqrt2\,c_Z}
=
0.9349126180\ldots,
\\
v_{\rm gate}^{\rm CCZ}
&=
1-\frac{1}{12\sqrt2\,c_{\rm CCZ}}
=
0.9148402794\ldots .
\end{aligned}
\label{eq:cczvisibilitythresholds}
\end{equation}
The CCZ-specific analysis therefore lowers both thresholds relative to the
uniform Boolean-phase bound.  The larger gate thresholds reflect the
composition of two imperfect state certificates through the Bures-angle
bound.

\bibliography{main}

@article{Bell1964,
  author = {Bell, John S.},
  title = {On the {Einstein Podolsky Rosen} paradox},
  journal = {Physics Physique Fizika},
  volume = {1},
  number = {3},
  pages = {195--200},
  year = {1964},
  doi = {10.1103/PhysicsPhysiqueFizika.1.195}
}

@article{CHSH1969,
  author = {Clauser, John F. and Horne, Michael A. and Shimony, Abner and Holt, Richard A.},
  title = {Proposed experiment to test local hidden-variable theories},
  journal = {Phys. Rev. Lett.},
  volume = {23},
  number = {15},
  pages = {880--884},
  year = {1969},
  doi = {10.1103/PhysRevLett.23.880}
}

@article{Tsirelson1980,
  author = {Cirel'son, Boris S.},
  title = {Quantum generalizations of {Bell}'s inequality},
  journal = {Lett. Math. Phys.},
  volume = {4},
  number = {2},
  pages = {93--100},
  year = {1980},
  doi = {10.1007/BF00417500}
}

@article{PopescuRohrlich1992,
  author  = {Popescu, Sandu and Rohrlich, Daniel},
  title   = {Which states violate {Bell}'s inequality maximally?},
  journal = {Physics Letters A},
  volume  = {169},
  number  = {6},
  pages   = {411--414},
  year    = {1992},
  doi     = {10.1016/0375-9601(92)90819-8}
}

@article{MayersYao2004,
  author = {Mayers, Dominic and Yao, Andrew},
  title = {Self testing quantum apparatus},
  journal = {Quantum Inf. Comput.},
  volume = {4},
  number = {4},
  pages = {273--286},
  year = {2004},
}

@article{McKagueYangScarani2012,
  author = {McKague, Matthew and Yang, Tzyh Haur and Scarani, Valerio},
  title = {Robust self-testing of the singlet},
  journal = {J. Phys. A: Math. Theor.},
  volume = {45},
  number = {45},
  pages = {455304},
  year = {2012},
  doi = {10.1088/1751-8113/45/45/455304},
}

@article{ReichardtUngerVazirani2013,
  author = {Reichardt, Ben W. and Unger, Falk and Vazirani, Umesh},
  title = {Classical command of quantum systems},
  journal = {Nature},
  volume = {496},
  number = {7446},
  pages = {456--460},
  year = {2013},
  doi = {10.1038/nature12035},
}

@article{SupicBowles2020,
  author = {{\v{S}}upi{\'c}, Ivan and Bowles, Joseph},
  title = {Self-testing of quantum systems: a review},
  journal = {Quantum},
  volume = {4},
  pages = {337},
  year = {2020},
  doi = {10.22331/q-2020-09-30-337},
}

@article{BalanzoJuando2026,
  author = {Balanz{\'o}-Juand{\'o}, Maria and Coladangelo, Andrea and Augusiak, Remigiusz and Ac{\'i}n, Antonio and {\v S}upi{\'c}, Ivan},
  title = {All pure multipartite entangled states of qubits can be self-tested},
  journal = {Nat. Commun.},
  volume = {17},
  pages = {4463},
  year = {2026},
  doi = {10.1038/s41467-026-70829-x},
}

@misc{Liu2026Scalable,
  author = {Liu, Jinchang and Huber, Elias X. and Du, Zhenyu and Zhang, Xingjian and Ma, Xiongfeng},
  title = {Scalable self-testing of generic multipartite quantum states},
  year = {2026},
  eprint = {2605.15106},
  archiveprefix = {arXiv},
  primaryclass = {quant-ph},
  doi = {10.48550/arXiv.2605.15106}
}

@article{YangVertesiBancal2014,
  author = {Yang, Tzyh Haur and V{\'e}rtesi, Tam{\'a}s and Bancal, Jean-Daniel and Scarani, Valerio and Navascu{\'e}s, Miguel},
  title = {Robust and versatile black-box certification of quantum devices},
  journal = {Phys. Rev. Lett.},
  volume = {113},
  number = {4},
  pages = {040401},
  year = {2014},
  doi = {10.1103/PhysRevLett.113.040401},
}

@article{Bancal2015,
  author = {Bancal, Jean-Daniel and Navascu{\'e}s, Miguel and Scarani, Valerio and V{\'e}rtesi, Tam{\'a}s and Yang, Tzyh Haur},
  title = {Physical characterization of quantum devices from nonlocal correlations},
  journal = {Phys. Rev. A},
  volume = {91},
  number = {2},
  pages = {022115},
  year = {2015},
  doi = {10.1103/PhysRevA.91.022115},
}

@article{BampsPironio2015,
  author = {Bamps, C{\'e}dric and Pironio, Stefano},
  title = {Sum-of-squares decompositions for a family of {CHSH}-like inequalities and their application to self-testing},
  journal = {Phys. Rev. A},
  volume = {91},
  number = {5},
  pages = {052111},
  year = {2015},
  doi = {10.1103/PhysRevA.91.052111},
}

@article{Kaniewski2016,
  author = {Kaniewski, J{\k{e}}drzej},
  title = {Analytic and nearly optimal self-testing bounds for the {Clauser--Horne--Shimony--Holt} and {Mermin} inequalities},
  journal = {Phys. Rev. Lett.},
  volume = {117},
  number = {7},
  pages = {070402},
  year = {2016},
  doi = {10.1103/PhysRevLett.117.070402},
}

@article{WangWuScarani2016,
  author = {Wang, Yukun and Wu, Xingyao and Scarani, Valerio},
  title = {All the self-testings of the singlet for two binary measurements},
  journal = {New J. Phys.},
  volume = {18},
  number = {2},
  pages = {025021},
  year = {2016},
  doi = {10.1088/1367-2630/18/2/025021},
}

@inproceedings{McKague2011,
  author = {McKague, Matthew},
  title = {Self-testing graph states},
  booktitle = {Theory of Quantum Computation, Communication, and Cryptography},
  editor = {Bacon, Dave and Martin-Delgado, Miguel and Roetteler, Martin},
  series = {Lecture Notes in Computer Science},
  volume = {6745},
  pages = {104--120},
  publisher = {Springer},
  year = {2014},
  doi = {10.1007/978-3-642-54429-3_7},
}

@article{Wu2014,
  author = {Wu, Xingyao and Cai, Yu and Yang, Tzyh Haur and Le, Huy Nguyen and Bancal, Jean-Daniel and Scarani, Valerio},
  title = {Robust self-testing of the three-qubit {$W$} state},
  journal = {Phys. Rev. A},
  volume = {90},
  number = {4},
  pages = {042339},
  year = {2014},
  doi = {10.1103/PhysRevA.90.042339},
}

@article{Li2018,
  author = {Li, Xinhui and Cai, Yu and Han, Yunguang and Wen, Qiaoyan and Scarani, Valerio},
  title = {Self-testing using only marginal information},
  journal = {Phys. Rev. A},
  volume = {98},
  number = {5},
  pages = {052331},
  year = {2018},
  doi = {10.1103/PhysRevA.98.052331},
}

@article{Li2020,
  author = {Li, Xinhui and Wang, Yukun and Han, Yunguang and Qin, Su-Juan and Gao, Fei and Wen, Qiaoyan},
  title = {Self-testing of symmetric three-qubit states},
  journal = {IEEE J. Sel. Areas Commun.},
  volume = {38},
  number = {3},
  pages = {589--597},
  year = {2020},
  doi = {10.1109/JSAC.2020.2968994},
}

@article{Coladangelo2017,
  author = {Coladangelo, Andrea and Goh, Koon Tong and Scarani, Valerio},
  title = {All pure bipartite entangled states can be self-tested},
  journal = {Nat. Commun.},
  volume = {8},
  pages = {15485},
  year = {2017},
  doi = {10.1038/ncomms15485},
}

@article{Baccari2020,
  author = {Baccari, Flavio and Augusiak, Remigiusz and {\v{S}}upi{\'c}, Ivan and Tura, Jordi and Ac{\'i}n, Antonio},
  title = {Scalable {Bell} inequalities for qubit graph states and robust self-testing},
  journal = {Phys. Rev. Lett.},
  volume = {124},
  number = {2},
  pages = {020402},
  year = {2020},
  doi = {10.1103/PhysRevLett.124.020402},
}

@article{Jamiolkowski1972,
  author = {Jamio{\l}kowski, Andrzej},
  title = {Linear transformations which preserve trace and positive semidefiniteness of operators},
  journal = {Rep. Math. Phys.},
  volume = {3},
  number = {4},
  pages = {275--278},
  year = {1972},
  doi = {10.1016/0034-4877(72)90011-0}
}

@article{Choi1975,
  author = {Choi, Man-Duen},
  title = {Completely positive linear maps on complex matrices},
  journal = {Linear Algebra Appl.},
  volume = {10},
  number = {3},
  pages = {285--290},
  year = {1975},
  doi = {10.1016/0024-3795(75)90075-0}
}

@article{DallArno2017,
  author = {Dall'Arno, Michele and Brandsen, Sarah and Buscemi, Francesco},
  title = {Device-independent tests of quantum channels},
  journal = {Proc. R. Soc. A},
  volume = {473},
  number = {2200},
  pages = {20160721},
  year = {2017},
  doi = {10.1098/rspa.2016.0721},
}

@article{Sekatski2018,
  author = {Sekatski, Pavel and Bancal, Jean-Daniel and Wagner, Sebastian and Sangouard, Nicolas},
  title = {Certifying the building blocks of quantum computers from {Bell}'s theorem},
  journal = {Phys. Rev. Lett.},
  volume = {121},
  number = {18},
  pages = {180505},
  year = {2018},
  doi = {10.1103/PhysRevLett.121.180505},
}

@article{Wagner2020,
  author = {Wagner, Sebastian and Bancal, Jean-Daniel and Sangouard, Nicolas and Sekatski, Pavel},
  title = {Device-independent characterization of quantum instruments},
  journal = {Quantum},
  volume = {4},
  pages = {243},
  year = {2020},
  doi = {10.22331/q-2020-03-19-243},
}

@article{Sekatski2023,
  author = {Sekatski, Pavel and Bancal, Jean-Daniel and Ioannou, Marie and Afzelius, Mikael and Brunner, Nicolas},
  title = {Toward the device-independent certification of a quantum memory},
  journal = {Phys. Rev. Lett.},
  volume = {131},
  number = {17},
  pages = {170802},
  year = {2023},
  doi = {10.1103/PhysRevLett.131.170802},
}

@article{Bock2024,
  author = {Bock, Matthias and Sekatski, Pavel and Bancal, Jean-Daniel and Kucera, Stephan and Bauer, Tobias and Sangouard, Nicolas and Becher, Christoph and Eschner, J{\"u}rgen},
  title = {Calibration-independent bound on the unitarity of a quantum channel with application to a frequency converter},
  journal = {npj Quantum Inf.},
  volume = {10},
  pages = {63},
  year = {2024},
  doi = {10.1038/s41534-024-00859-0}
}

@article{Zeng2020,
  author = {Zeng, Pei and Zhou, You and Liu, Zhenhuan},
  title = {Quantum gate verification and its application in property testing},
  journal = {Phys. Rev. Res.},
  volume = {2},
  number = {2},
  pages = {023306},
  year = {2020},
  doi = {10.1103/PhysRevResearch.2.023306},
}

@article{ZhuZhang2020,
  author = {Zhu, Huangjun and Zhang, Haoyu},
  title = {Efficient verification of quantum gates with local operations},
  journal = {Phys. Rev. A},
  volume = {101},
  number = {4},
  pages = {042316},
  year = {2020},
  doi = {10.1103/PhysRevA.101.042316},
}

@article{Zhang2022,
  author = {Zhang, Rui-Qi and Hou, Zhibo and Tang, Jun-Feng and Shang, Jiangwei and Zhu, Huangjun and Xiang, Guo-Yong and Li, Chuan-Feng and Guo, Guang-Can},
  title = {Efficient experimental verification of quantum gates with local operations},
  journal = {Phys. Rev. Lett.},
  volume = {128},
  number = {2},
  pages = {020502},
  year = {2022},
  doi = {10.1103/PhysRevLett.128.020502}
}

@article{Chiribella2008,
  author = {Chiribella, Giulio and D'Ariano, Giacomo Mauro and Perinotti, Paolo},
  title = {Transforming quantum operations: quantum supermaps},
  journal = {Europhys. Lett.},
  volume = {83},
  number = {3},
  pages = {30004},
  year = {2008},
  doi = {10.1209/0295-5075/83/30004},
}

@article{Sarkar2026,
  author = {Sarkar, Shubhayan},
  title = {Any Unitary Gate Can Be Certified Device-Independently in a Quantum Network},
  journal = {Phys. Rev. Lett.},
  volume = {137},
  number = {3},
  pages = {030802},
  year = {2026},
  doi = {10.1103/m1tx-9mx1},
}

@misc{Barizien2026,
  author = {Barizien, Victor and Branciard, Cyril and Abbott, Alastair A. and Bancal, Jean-Daniel and Sekatski, Pavel},
  title = {Self-testing quantum supermaps},
  year = {2026},
  eprint = {2606.25124},
  archiveprefix = {arXiv},
  primaryclass = {quant-ph},
  doi = {10.48550/arXiv.2606.25124}
}

@article{Renou2018,
  author = {Renou, Marc-Olivier and Kaniewski, J{\k{e}}drzej and Brunner, Nicolas},
  title = {Self-testing entangled measurements in quantum networks},
  journal = {Phys. Rev. Lett.},
  volume = {121},
  number = {25},
  pages = {250507},
  year = {2018},
  doi = {10.1103/PhysRevLett.121.250507},
}

@article{Bancal2018,
  author = {Bancal, Jean-Daniel and Sangouard, Nicolas and Sekatski, Pavel},
  title = {Noise-resistant device-independent certification of {Bell}-state measurements},
  journal = {Phys. Rev. Lett.},
  volume = {121},
  number = {25},
  pages = {250506},
  year = {2018},
  doi = {10.1103/PhysRevLett.121.250506},
}

@article{Supic2023,
  author = {{\v{S}}upi{\'c}, Ivan and Bowles, Joseph and Renou, Marc-Olivier and Ac{\'i}n, Antonio and Hoban, Matty J.},
  title = {Quantum networks self-test all entangled states},
  journal = {Nat. Phys.},
  volume = {19},
  number = {5},
  pages = {670--675},
  year = {2023},
  doi = {10.1038/s41567-023-01945-4},
}

@article{Tavakoli2022,
  author = {Tavakoli, Armin and Pozas-Kerstjens, Alejandro and Luo, Ming-Xing and Renou, Marc-Olivier},
  title = {{Bell} nonlocality in networks},
  journal = {Rep. Prog. Phys.},
  volume = {85},
  number = {5},
  pages = {056001},
  year = {2022},
  doi = {10.1088/1361-6633/ac41bb},
}

@article{Rossi2013,
  author = {Rossi, Matteo and Huber, Marcus and Bru{\ss}, Dagmar and Macchiavello, Chiara},
  title = {Quantum hypergraph states},
  journal = {New J. Phys.},
  volume = {15},
  number = {11},
  pages = {113022},
  year = {2013},
  doi = {10.1088/1367-2630/15/11/113022},
}

@article{Gachechiladze2016,
  author = {Gachechiladze, Mariami and Budroni, Costantino and G{\"u}hne, Otfried},
  title = {Extreme violation of local realism in quantum hypergraph states},
  journal = {Phys. Rev. Lett.},
  volume = {116},
  number = {7},
  pages = {070401},
  year = {2016},
  doi = {10.1103/PhysRevLett.116.070401},
}

@article{ZhuHayashi2019,
  author = {Zhu, Huangjun and Hayashi, Masahito},
  title = {Efficient verification of hypergraph states},
  journal = {Phys. Rev. Appl.},
  volume = {12},
  number = {5},
  pages = {054047},
  year = {2019},
  doi = {10.1103/PhysRevApplied.12.054047},
}

@article{Takeuchi2019,
  author = {Takeuchi, Yuki and Morimae, Tomoyuki and Hayashi, Masahito},
  title = {Quantum computational universality of hypergraph states with {Pauli-X} and {Z}-basis measurements},
  journal = {Sci. Rep.},
  volume = {9},
  pages = {13585},
  year = {2019},
  doi = {10.1038/s41598-019-49968-3},
}

@article{MillerMiyake2016,
  author = {Miller, Jacob and Miyake, Akimasa},
  title = {Hierarchy of universal entanglement in two-dimensional measurement-based quantum computation},
  journal = {npj Quantum Inf.},
  volume = {2},
  pages = {16036},
  year = {2016},
  doi = {10.1038/npjqi.2016.36},
}

@article{ChenYanZhou2024,
  author = {Chen, Junjie and Yan, Yuxuan and Zhou, You},
  title = {Magic of quantum hypergraph states},
  journal = {Quantum},
  volume = {8},
  pages = {1351},
  year = {2024},
  doi = {10.22331/q-2024-05-21-1351},
}

@article{Jones2013,
  author = {Jones, Cody},
  title = {Low-overhead constructions for the fault-tolerant {Toffoli} gate},
  journal = {Phys. Rev. A},
  volume = {87},
  number = {2},
  pages = {022328},
  year = {2013},
  doi = {10.1103/PhysRevA.87.022328},
}

@article{Eastin2013,
  author = {Eastin, Bryan},
  title = {Distilling one-qubit magic states into {Toffoli} states},
  journal = {Phys. Rev. A},
  volume = {87},
  number = {3},
  pages = {032321},
  year = {2013},
  doi = {10.1103/PhysRevA.87.032321},
}

@misc{BooleanPhaseCode2026,
  author       = {Han, Yunguang},
  year         = {2026},
    howpublished = {GitHub},
note = {Available at \url {https://github.com/asstger/boolean-phase-gate-certification}}
}

@article{Huang2024,
  author  = {Huang, Jieshan and Li, Xudong and Chen, Xiaojiong and Zhai, Chonghao and Zheng, Yun and Chi, Yulin and Li, Yan and He, Qiongyi and Gong, Qihuang and Wang, Jianwei},
  title   = {Demonstration of hypergraph-state quantum information processing},
  journal = {Nature Communications},
  volume  = {15},
  number  = {1},
  pages   = {2601},
  year    = {2024},
  doi     = {10.1038/s41467-024-46830-7}
}

@misc{Han2026CCZState,
  author = {Han, Yunguang and Kuang, Xin and Bu, Xingyuan and Wang, Yukun},
  title = {Device-Independent Self-Testing of the Three-Qubit {CCZ} Hypergraph State},
  year = {2026},
  eprint = {2607.21288},
  archiveprefix = {arXiv},
  primaryclass = {quant-ph},
  doi = {10.48550/arXiv.2607.21288}
}

\end{document}